\documentclass{amsart}
\usepackage{times}
\usepackage{graphicx}
\usepackage[margin=1in]{geometry}
\usepackage{amsmath,amsfonts,amssymb, amsthm,mathrsfs}
\usepackage{bm}
\usepackage{setspace}
\usepackage[ansinew]{inputenc}
\usepackage{multicol}
\usepackage{color}
\usepackage{url}

\newcommand{\diag}{{\rm diag}}
\newcommand{\Diag}{{\rm Diag}}
\newcommand{\proj}{{\rm proj}}
\newcommand{\Isom}{{\rm Isom}}
\newcommand{\isoR}{{\rm Isom}(\mathbb{R}^n)}
\newcommand{\isoC}{{\rm Isom}(\mathbb{C}^n)}
\newcommand{\dist}{{\rm dist}}
\newcommand{\K}{{\mathcal K}}
\newcommand{\KK}{{\widehat{K}}}

\newtheorem{thm}{Theorem}
\newtheorem{remark}[thm]{Remark}
\newtheorem{question}[thm]{Question}
\newtheorem{lemma}[thm]{Lemma}

\newtheorem{proposition}[thm]{Proposition}
\newtheorem{definition}[thm]{Definition}

\begin{document}

\title[Two-Point Correlation Functions]{Two-Point Correlation Functions and Universality for the Zeros of Systems of SO(n+1)-invariant Gaussian Random Polynomials}

\author[P. ~Bleher]{Pavel M. Bleher$^{1}$}

\author[Y. ~Homma]{Yushi Homma$^{1,2}$}

\author[R. ~Roeder]{Roland K.\ W.\ Roeder$^{1}$}

\footnotetext[1]{Department of Mathematical Sciences, Indiana University-Purdue University Indianapolis, 402 N. Blackford St., Indianapolis, IN 46202, USA}
\footnotetext[2]{Carmel High School, 520 E. Main St., Carmel, IN 46032, USA.}
\footnotetext[3]{Stanford University, Palo Alto, CA, United States.}

\begin{abstract}
We study the two-point correlation functions for the zeroes of systems of
$SO(n+1)$-invariant Gaussian random polynomials on $\mathbb{RP}^n$ and systems
of $\isoR$-invariant Gaussian analytic functions. Our result reflects the
same ``repelling," ``neutral," and ``attracting" short-distance asymptotic
behavior, depending on the dimension, as was discovered in the complex case by
Bleher, Shiffman, and Zelditch. 

We then prove that the correlation function for the $\isoR$-invariant Gaussian analytic functions is ``universal,'' describing the scaling limit
of the correlation function for the restriction of systems of the $SO(k+1)$-invariant Gaussian random polynomials to any $n$-dimensional $C^2$ submanifold $M \subset \mathbb{RP}^k$.  This provides a real counterpart to the universality results that were proved in the complex case by Bleher, Shiffman, and Zelditch.
\end{abstract}

\maketitle

\section{Introduction}

This paper concerns the $SO(n+1)$-invariant ensemble of Gaussian random
polynomials on $\mathbb{RP}^n$ and the $\Isom(\mathbb{R}^n)$-invariant ensemble
of Gaussian random analytic functions on $\mathbb{R}^n$.
The $SO(n+1)$-invariant ensemble consists of random polynomials of the form:
\begin{equation} \label{GRP}
F(\bm{X}) :=\sum\limits_{|\bm{\alpha}| = d} \sqrt{\binom{d}{\bm{\alpha}}} a_{\bm{\alpha}} \bm{X}^{\bm{\alpha}},
\end{equation}
where $\bm{X}\in\mathbb{R}^{n+1}$ and the $a_{\bm{\alpha}}$ are independent and identically distributed (iid) on the standard normal distribution, $\mathcal{N}\left(0, 1\right)$. Here, we use the following multi-index notation: for any $\bm{\alpha} \in  \left(\mathbb{Z}_{\geq 0}\right)^{n+1}$, one defines: \begin{equation}\bm{X}^{\bm{\alpha}}:=\prod\limits_{i=1}^{n+1} X_i^{\alpha_i}, \quad |\bm{\alpha} | = \sum\limits_{i = 1}^{n+1} \alpha_i \quad \mbox{and}\quad \binom{d}{\bm{\alpha}}  = \frac{d!}{\prod\limits_{j=1}^{n+1} \alpha_j!}.\end{equation}
We will study the simultaneous zeroes on the projective space $\mathbb{RP}^n$ of the systems: 
\begin{equation}\label{System} 
\bm{F}: \mathbb{R}^{n+1} \rightarrow \mathbb{R}^n \qquad \mbox{where} \qquad
\bm{F}=\left(F_1(\bm{X}), F_2(\bm{X}), \dots,  F_n(\bm{X})\right),
\end{equation}
 where each $F_i$ is an independently chosen random function of
the form in Equation~(\ref{GRP}). Almost surely, the common zero set of
$\bm{F}$ will be finitely many points. We equip $\mathbb{RP}^n$ with the
Riemannian metric obtained from its double cover by the unit sphere
$\mathbb{S}^n \subset \mathbb{R}^{n+1}$. The simultaneous zeroes of ensemble
$(\ref{System})$ are invariant under the isometries by elements of $SO(n+1)$;
see Section \ref{Invariance}. Because of this symmetry, authors have described
this ensemble as the ``most natural" ensemble of a random polynomials defined
on $\mathbb{RP}^n$. For this reason, it has been extensively studied by
Kostlan-Edelman\cite{Kostlan}, Shub-Smale\cite{ShubSmale}, and others.

The $\Isom(\mathbb{R}^n)$-invariant ensemble of Gaussian random analytic functions
is defined by the following:
\begin{equation}\label{GAF}
\bm{f}: \mathbb{R}^{n} \rightarrow \mathbb{R}^n \quad \mbox{where} \quad
\bm{f}=\left(f_1(\bm{x}), f_2(\bm{x}), \dots,  f_n(\bm{x})\right), \quad
\mbox{with} \quad f_i(\bm{x}):=\sum\limits_{\bm{\alpha}}
\frac{a_{\bm{\alpha}}
\bm{x}^{\bm{\alpha}}}{\sqrt{\bm{\alpha}!}},
\end{equation}
where $a_{\alpha}$ are iid on the standard normal
distribution, $\mathcal{N}\left(0, 1\right)$. We will show in Section
\ref{Invariance} that the zeroes of this ensemble are invariant under all
isometries of $\mathbb{R}^n$.   We will see shortly that this ensemble is intimately tied to the $SO(n+1)$-invariance ensemble in the scaling limit as the degree $d \rightarrow \infty$.

The {\em probability density} of the zeros of the system (\ref{System}) at $\bm{x} \in \mathbb{RP}^n$ is defined to be 
\begin{align}\label{DEF_DENSITY}
\rho(\bm{x})=\lim\limits_{\delta \rightarrow 0} \frac{1}{{\rm Vol}\left(N_\delta({\bm{x}})\right)} \text{Pr}\left(\exists \text{ a zero of $\bm{F}$ in } N_\delta({\bm{x}})\right),
\end{align}
where $\begin{displaystyle} N_\delta({\bm{x}}):=\{\bm{\bm{y}}\in\mathbb{RP}^n\ :\ \mathrm{dist}(\bm{x},\bm{y})<\delta \}\end{displaystyle}$. It follows from the invariance that this ensemble (\ref{System}) has a constant density of zeroes given by 
\begin{equation}\label{Density_FINITE_D} 
\rho_d(\bm{x}) = \pi^{-\frac{n+1}{2}} \Gamma \left(\frac{n+1}{2}\right) d^{\frac{n}{2}}; 
\end{equation} see, for example, \cite[Sec. 7.2]{Kostlan}. Note: the volume of the real projective space is $ \pi^{\frac{n+1}{2}} \Gamma \left(\frac{n+1}{2}\right)^{-1}$, so the expected number of zeroes is simply $d^{\frac{n}{2}}$.
The analogous definition applies to the ensemble (\ref{GAF}) which, because of the invariance under isometries of $\mathbb{R}^n$ has 
constant density
\begin{equation}  \label{Density} \rho(\bm{x}) =  \pi^{-\frac{n+1}{2}} \Gamma\left(\frac{n+1}{2}\right).
\end{equation}

The correlation function between the zeros of the system (\ref{System}) at the two points $\bm{x}$ and $\bm{y}$ in $\mathbb{RP}^n$ is defined to be
\begin{equation}
\label{DefK}K_{n,d}(\bm{x},\bm{y}):=\lim\limits_{\delta\rightarrow 0}  \ \frac{\text{Pr}\left(\exists\text{ a zero of } \bm{F} \text{ in } N_\delta({\bm{x}})\text{ and } \exists \text{ a zero of }\bm{F} \text{ in } N_\delta({\bm{y}}) \right)}{\text{Pr}\left(\exists\text{ a zero of } \bm{F} \text{ in } N_\delta({\bm{x}}) \right) \text{Pr}\left(\exists \text{ a zero of }\bm{F} \text{ in } N_\delta({\bm{y}}) \right)}.
\end{equation}
It follows from the $SO(n+1)$ invariance that $K_{n,d}(\bm{x},\bm{y})$
depends only on the distance between $\bm{x}$ and $\bm{y}$.  For this
reason, we can write $K_{n,d}(\bm{x},\bm{y}) \equiv
K_{n,d}(t)$, where $t = \dist_{\mathbb{RP}^n}(\bm{x},\bm{y})$.
Similarly,
for any $\bm{x},\bm{y} \in \mathbb{R}^n$, the two point
correlation function $\mathcal{K}_n(\bm{x},\bm{y})$ between zeros of
(\ref{GAF}) depends only on $\dist_{\mathbb{R}^n}(\bm{x},\bm{y})$.  We have

\begin{thm} \label{CORR_FINITE_D} 
For any $\bm{x} \neq \bm{y} \in \mathbb{RP}^n$, let $t = \dist_{\mathbb{RP}^n}(\bm{x},\bm{y})$.
For fixed $d \geq 3$, the correlation function between zeros of the $SO(n+1)$-invariant ensemble satisfies 
\begin{align}\label{SHORT_RANGE_D}
K_{n,d}(\bm{x},\bm{y}) \equiv K_{n,d}(t) = A_{n,d} \ t^{2-n} + O\left(t^{3-n}\right) \quad \mbox{as} \quad t \rightarrow 0, \quad  \mbox{where} \quad
A_{n,d} = \left(\frac{d-1}{d^{\frac{n}{2}}}\right) \frac{\sqrt{\pi} \Gamma\left(\frac{n+2}{2}\right)}{2\Gamma\left( \frac{n+1}{2}
\right)}.
\end{align}
\end{thm}

\begin{thm} \label{CORR_ISOM_INV}
For any $\bm{x} \neq \bm{y} \in \mathbb{R}^n$, let $t = \dist_{\mathbb{R}^n}(\bm{x},\bm{y})$.  The correlation function between zeros of the $\isoR$-invariant ensemble satisfies 
\begin{eqnarray}\label{SHORT_RANGE}
K_n(\bm{x},\bm{y}) \equiv K_n(t) = A_n \ t^{2-n} + O\left(t^{3-n}\right) \quad \mbox{as} \quad  t \rightarrow 0, \quad  \mbox{where} \quad
A_n = \frac{\sqrt{\pi} \Gamma\left(\frac{n+2}{2}\right)}{2\Gamma\left( \frac{n+1}{2}
\right)},
\end{eqnarray}
and
\begin{eqnarray} \label{LONG_RANGE}
K_n(t) = 1 + O\left(t \mathrm{e}^{-\frac{t^2}{2}}\right) \quad \mbox{as $t \rightarrow \infty$.}
\end{eqnarray}
\end{thm}

Given a $C^2$ submanifold $M \subset \mathbb{RP}^k$ having dimension $n$, the
restrictions of $n$ of the polynomials chosen iid from the $SO(k+1)$-invariant ensemble (\ref{GRP}) has a
well-defined zero set which again consists a.s. of finitely many points. 
We give $M \subset \mathbb{RP}^k$ the metric induced by the double
cover of $\mathbb{RP}^k$ by the unit sphere~$\mathbb{S}^k$.   More specifically,
we obtain a Riemannian metric on $M$ using the inclusion of tangent spaces $T_p M \subset T_p \mathbb{RP}^k$.
When restricted to a sufficiently small neighborhood of the origin, the
orthogonal projection $\proj_p : T_p M \rightarrow M$ provides a system of
local coordinates on $M$.  We will use these systems of local coordinates to
study the correlation between zeros of the restriction of the $SO(k+1)$
invariant ensemble to $M$.
The next theorem expresses $K_n(\bm{x},\bm{y})$ as the universal correlation function in the scaling limit $d\rightarrow \infty$ for the restriction of the $SO(k+1)$-invariant ensemble to $M \subset \mathbb{RP}^k$:

\begin{thm}\label{THM:UNIVERSALITY}
Let $M \subset \mathbb{RP}^k$ be a $C^2$ submanifold of dimension $n$ and
$K_{n,d,M}(\bm x, \bm y)$ denote the correlation function between zeros of $n$ polynomials chosen iid from the degree~$d$
$SO(k+1)$ invariant ensemble restricted to $M$.  Then, for any $p \in M$ and any $\bm{x}, \bm{y} \in T_p M$ we have
\begin{eqnarray*}
K_{n,d,M}\left(\proj_p\left(\frac{\bm x}{\sqrt{d}}\right),\proj_p \left(\frac{\bm y}{\sqrt{d}}
\right)\right) = K_n\big(\bm x, \bm y \big) + O\left(\frac{1}{\sqrt{d}}\right)  \quad \mbox{as $d \rightarrow \infty$.}
\end{eqnarray*}
The constant in the estimate is uniform on compact subsets of $T_p M \times T_p M \setminus \Diag$, where
$\Diag = \{(\bm{x},\bm{y}) \in T_p M \times T_p M\ \, : \, \bm{x} = \bm{y} \}$.
\end{thm}

Our techniques are largely based on those of Bleher and Di
\cite{BleherDi1}, who use the Kac-Rice formula (see Section \ref{Kac-Rice} below) to study the $n$-point correlation functions for
the $SO(1,1)$ and $SO(2)$-invariant polynomials in one variable. Moreover,
our results in the higher dimensional real case yield the exact same
short-distance asymptotic behavior (with a different constant) as those of
Bleher, Shiffman, and Zelditch \cite{BSZUniversality2,BSZUniversality,BSZ3} in the
complex case. 
These asymptotic behaviors can be interpreted as ``repelling" for
$n=1$, ``neutral" for $n=2$, and ``attracting" for $n\geq3$. See Figure
\ref{Graphs} for numerical plots of $K_n(t)$ for $n=1$,
$n=2$, and $n=3$.  

We remark that calculation of the leading order asymptotics is more delicate in
the real case than in the complex case because one cannot apply Wick's Theorem
to the real Kac-Rice formula.  Similar types of analysis have been done in the
real setting by Nicolaescu, who studied the critical points for random Fourier
series.  It is interesting that he found the same exponent of $2-n$ arising in his
work \cite[Eqn. 1.15]{Nicolaescu}.

Theorem \ref{THM:UNIVERSALITY} above provides a real analog of the
universality results that were obtained in the complex setting by Bleher, Shiffman, and Zelditch
\cite{BSZUniversality2,BSZUniversality}.  Thus, the plots shown in Figure
\ref{Graphs} depict the universal scaling limits of the correlation functions
for any submanifold $M \subset \mathbb{RP}^k$ of dimension $1, 2,$ or $3$.

The scaling limit used in Theorem \ref{THM:UNIVERSALITY} is needed to get a
universal correlation function.  This is illustrated in
Section~\ref{SEC:GEOMETRY_FINITE_D} where we show that when restricted to a
parabola $y = b x^2$ the leading term from the correlation between zeros for
the $SO(3)$-invariant polynomials of degree $3$ near $x=0$ depends
non-trivially on $b$.  More generally, it can be interesting to ask how the
geometry of $M$ affects the correlation function $K$ for finite degree $d$.

\vspace{0.1in}

The proof of Theorem \ref{THM:UNIVERSALITY} easily adapts to to
complex setting: The $SU(k+1)$-invariant ensemble of polynomials are obtained by
interpreting the variables in (\ref{GRP}) as complex and replacing the real
Gaussians $a_{\bm \alpha}$ with complex Gaussians.  The $\isoC$-invariant ensemble of Gaussian analytic functions on $\mathbb{C}^m$ is obtained by making the same adaptations to (\ref{GAF}).  We obtain:

\begin{thm}\label{THM:UNIVERSALITYC}
Let $M \subset \mathbb{CP}^k$ be an complex analytic submanifold of dimension $n$ and
$K_{n,d,M}(\bm x, \bm y)$ denote the correlation function between zeros of $n$ polynomials chosen iid from the degree~$d$
$SU(k+1)$ invariant ensemble restricted to $M$.  Then, for any $p \in M$ and any $\bm{x}, \bm{y} \in T_p M$ we have
\begin{eqnarray*}
K_{n,d,M}\left(\proj_p\left(\frac{\bm x}{\sqrt{d}}\right),\proj_p \left(\frac{\bm y}{\sqrt{d}}
\right)\right) = K_n\big(\bm x, \bm y \big) + O\left(\frac{1}{\sqrt{d}}\right) \quad \mbox{as $d \rightarrow \infty$.}
\end{eqnarray*}
The constant in the estimate is uniform on compact subsets of $T_p M \times T_p M \setminus \Diag$, where
$\Diag = \{(\bm{x},\bm{y}) \in T_p M \times T_p M\ \, : \, \bm{x} = \bm{y} \}$.
\end{thm}

\noindent
This serves as a weaker version of the results from
\cite{BSZUniversality2,BSZUniversality} in that $M$ is required to be
embedded in projective space (instead of being an arbitrary K\"ahler manifold),
the line bundle is the hyperplane bundle (corresponding to the
\hbox{$SU(n+1)$}-invariant ensemble), and only two-point correlation functions are
considered.  
On the other hand, in the work of \cite{BSZUniversality2,BSZUniversality} the
manifold $M$ is assumed to be compact.  No such assumption is made in Theorems
\ref{THM:UNIVERSALITY} and \ref{THM:UNIVERSALITYC}.  For example, they can be
applied at any smooth point $p$ of a singular projective variety.

\begin{figure}
\includegraphics[scale=.363]{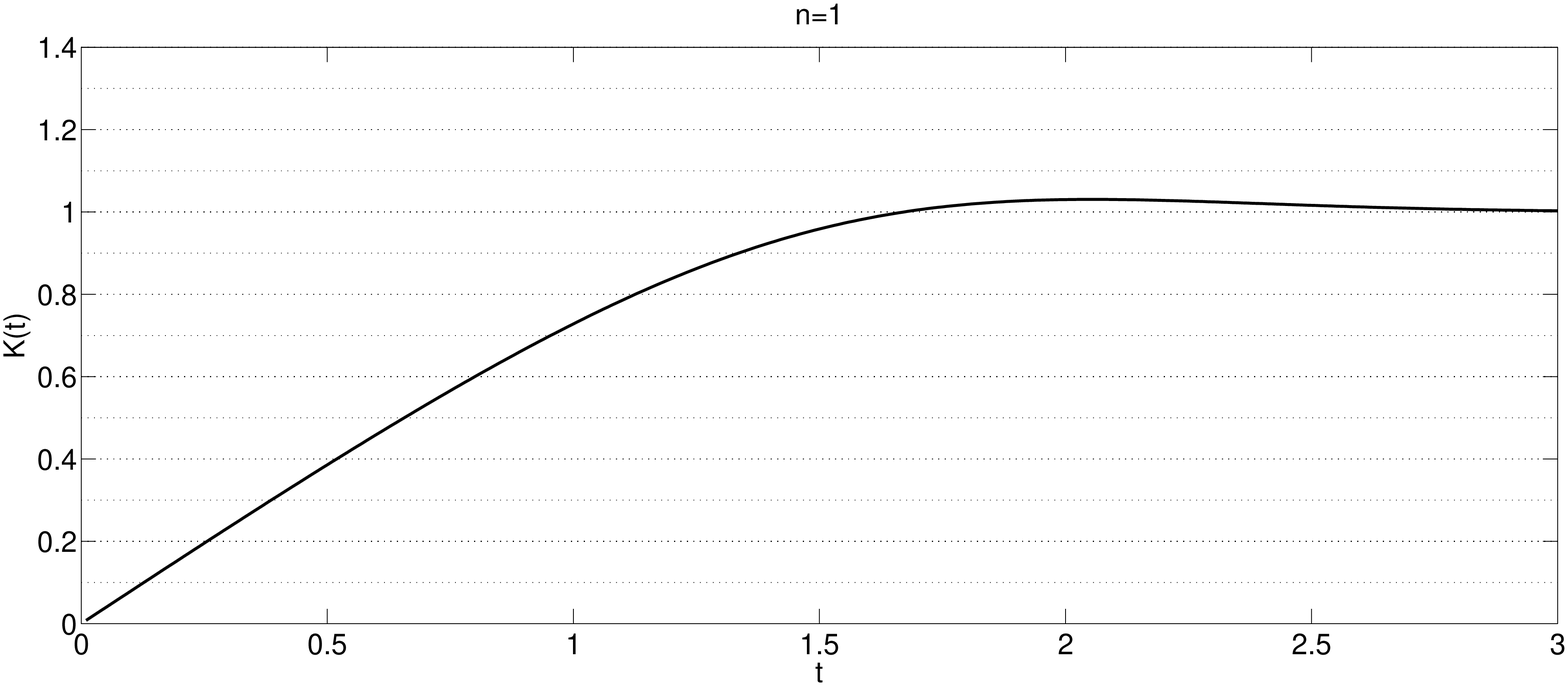}
\includegraphics[scale=.34]{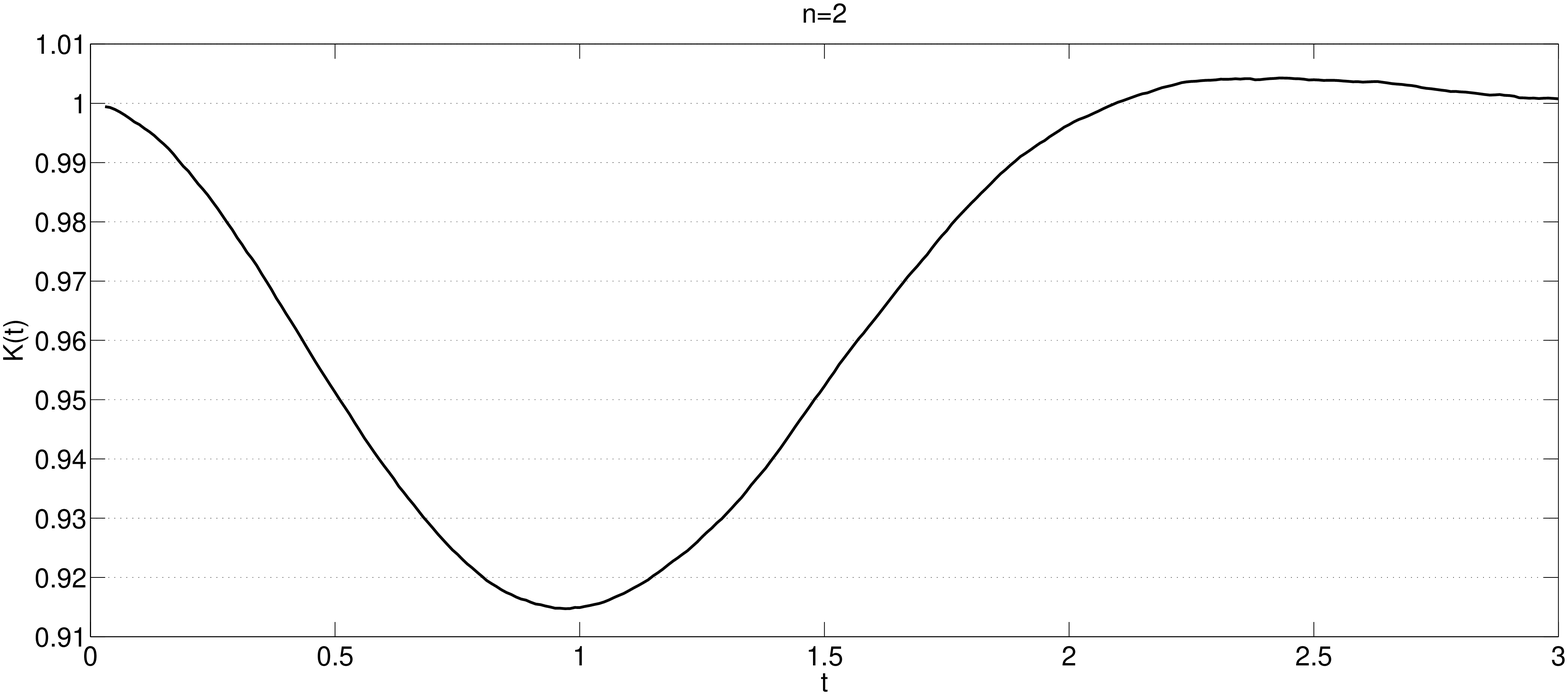}
\includegraphics[scale=.34]{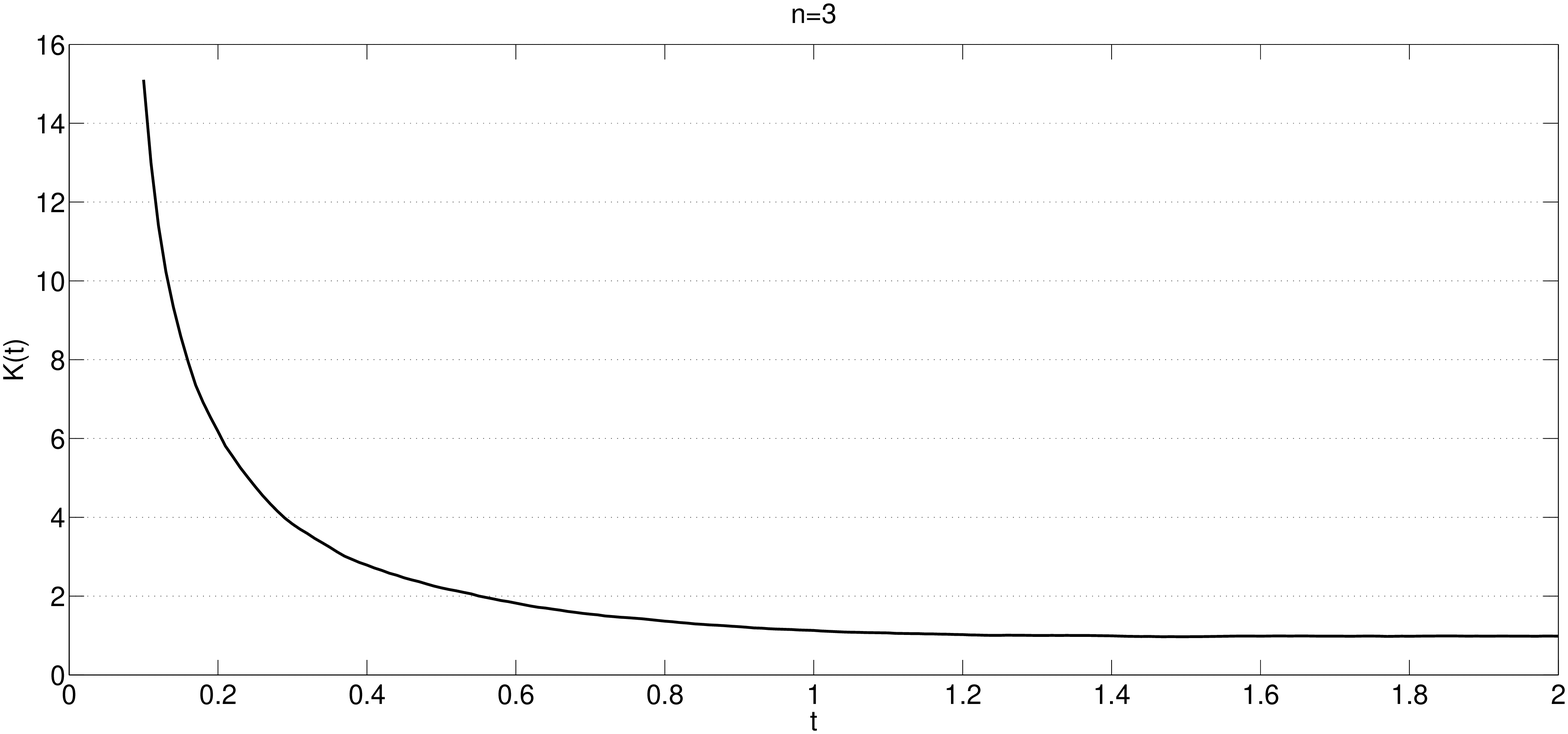}
\caption{Universal two-point limiting correlation functions $K_n\left(t\right)$ for $n = 1, 2, $ and 3, demonstrating the repelling, neutral and attracting behaviors. For $n=1$, the graph is obtained from Formula (5.35) in \cite{BleherDi1}. For $n=2$ and $n=3$, the graphs were computed using Monte Carlo integration applied to formula (\ref{Spherical}) with $10^7$ and $10^6$ points, respectively, for each $t$. The data was smoothed out by replacing each value with the average of it and the 14 nearest neighboring points. \label{Graphs}}
\end{figure}

\vspace{0.1in}

For general background on Gaussian random analytic functions and polynomials, we refer the reader to \cite{AT,AW,Hida,Peres,ShubSmale} and their references therein. Specifically to correlation functions, we refer the reader to the three papers listed above in the previous paragraph, as well as the works of Bogomolny, Bohigas, and Leboeuf \cite{Bogomolny1}, Tao and Vu \cite{TaoVu}, Bleher and Ridzal \cite{BleherRidzal}, and Bleher and Di \cite{BleherDi2}.

Our work fits in within the context of the emerging field ``random real
algebraic geometry.''  For example, Theorem~\ref{THM:UNIVERSALITY} applies to
the restriction of the $SO(k+1)$ ensemble to the smooth locus of a
real-algebraic subset of~$\mathbb{RP}^k$.  We refer the reader to the works of
Kostlan \cite{KOSTLAN1}, Shub-Smale \cite{ShubSmale}, Ibragimov-Podkorytov \cite{IP}, Burgisser
\cite{Burgisser}, Gayet-Welschinger \cite{GW1}, Ibragimov-Zaporozhets \cite{Zaporozhets}, Nastasescu \cite{Nastasescu},
Gayet-Welschinger \cite{GW2,GW3}, Lerario-Lundberg \cite{LerarioLundberg},
Gayet-Welschinger \cite{GW4,GW5}, and Fyodorov-Lerario-Lundberg \cite{FLL}.
\vspace{.1in}


The remainder of the paper will be organized as follows: In the following
Section \ref{Invariance}, we study the invariance properties of the ensembles
from (\ref{System}) and (\ref{GAF}).  We then use the invariance to reduce
Theorems \ref{CORR_FINITE_D} and \ref{CORR_ISOM_INV} to suitable versions in
affine coordinates (Theorem \ref{THM:LOCAL}). In Section \ref{Kac-Rice}, we
recall the Kac-Rice Formulae for the density and for the correlation
functions, the main tools used in our proof.  
In Section \ref{SEC:COVARIANCE} we compute the covariance matrices needed to prove Theorem \ref{THM:LOCAL}, as
well as their determinants, inverses, etc.
Theorem \ref{THM:LOCAL} consists of two statements (short-distance asymptotics and long-distance asymptotics),
which are proved in Sections \ref{SEC:SHORT_RANGE} and \ref{SEC:LONG_RANGE} respectively.
Section \ref{SEC:UNIVERSALITY} is dedicated to proving Theorem \ref{THM:UNIVERSALITY} about universality of the scaling limit.
Section~\ref{SEC:GEOMETRY_FINITE_D} provides an example showing that for finite degree the leading asymptotics depends on the
geometry of the submanifold $M \subset \mathbb{RP}^k$.
In Section~\ref{SEC:COMPLEX_CASE} we explain the changes that need to be made to the proof of Theorem \ref{THM:UNIVERSALITY} in order to prove the complex version, Theorem \ref{THM:UNIVERSALITYC}.

Appendix \ref{SEC:Lemma8} contains the proof of a general estimate which is used in Sections \ref{SEC:LONG_RANGE} and \ref{SEC:UNIVERSALITY}.
In Appendix \ref{kthMomentProof}, we prove a result
regarding the volume of random parallelotopes which is needed in Section \ref{SEC:SHORT_RANGE}.

\vspace{0.1in}
\noindent
{\bf Notation:}
Let $\diag_k\left(\bold{A}\right)$ denote the block-diagonal matrix with $k$ copies of the square matrix $\bold{A}$ along the diagonal.

\section{Invariance Properties and Reduction of Theorems  \ref{CORR_FINITE_D} and \ref{CORR_ISOM_INV} to local coordinates} \label{Invariance}

The $SO(n+1)$-invariant ensemble and the $\isoR$-invariant ensemble are instances of the following definition:

\begin{definition}
A {\em Gaussian analytic function} ${\bm h}: \mathbb{R}^n \rightarrow \mathbb{R}^m$ is an $m$-tuple 
$\left(h_1({\bm x}), h_2(\bm{x}), \dots,  h_m(\bm{x})\right)$ of functions $h_i: \mathbb{R}^n \rightarrow 
\mathbb{R}$ chosen iid of the form
\begin{equation}\label{General_GAF}
h_i(\bm{x}):=\sum\limits_{\bm{\alpha}} c_{\bm{\alpha}} a_{\bm{\alpha}}
\bm{x}^{\bm{\alpha}},
\end{equation}
where the $a_{\bm{\alpha}}$ are chosen iid on the standard normal
distribution $\mathcal{N}\left(0, 1\right)$ and the coefficients $c_{\bm{\alpha}}$  are chosen so that
$\sum\limits_{\bm{\alpha}} c_{\bm{\alpha}} \bm{x}^{\bm{\alpha}}$
converges for all $ x \in \mathbb{R}^n$.
\end{definition}

\begin{lemma}\label{LEM_AS_CONV}
A Gaussian analytic function $\bm{h}: \mathbb{R}^n \rightarrow \mathbb{R}^m$ almost
surely converges uniformly on compact subsets of 
$\mathbb{R}^n$ and moreover is real analytic on $\mathbb{R}^n$.
\end{lemma}

\begin{proof}
The proof of Lemma 2.2.3 from \cite{Peres} applies to show (\ref{General_GAF}) almost surely converges uniformly on compact subsets of $\mathbb{C}^n$ and hence defines a random complex analytic function on $\mathbb{C}^n$. By restricting the resulting functions to $\mathbb{R}^n$, we obtain the desired result.
\end{proof}

In particular, Lemma \ref{LEM_AS_CONV} justifies our consideration the $\isoR$-invariant ensemble  (\ref{GAF}) as actually defining a random function.  The following two lemmas justify our terminology ``$SO(n+1)$ invariant ensemble'' and ``$\isoR$-invariant ensemble'':

\begin{lemma} \label{InvarianceSOn}
The zeroes of the system $\bm{F}$ given in (\ref{System}) are invariant under
the action of $SO(n+1)$.  That is, for any open set $U\subset \mathbb{RP}^n$
and any  $A \in SO(n+1)$, we have $\mbox{Pr}\left(\bm{F} \mbox{ has a zero in }
U \right) = \mbox{Pr}\left(\bm{F} \mbox{ has a zero in } A\left(U\right)
\right).$
\end{lemma}
\begin{proof}
Each $F_i(\bm{X})$ defines a Gaussian process on $\mathbb{R}^{n+1}$, with mean 0 and covariance function
\begin{equation}
E(F_i(\bm{X})F_i(\bm{Y})) = \sum\limits_{\bm{\alpha} = d} \binom{d}{\bm{\alpha}} \bm{X}^{\bm{\alpha}}\bm{Y}^{\bm{\alpha}} = (\bm{X} \cdot \bm{Y})^d.
\end{equation}
Since any Gaussian process is uniquely determined by its first and second moments \cite[Theorem~2.1]{Hida}, this process is invariant under $SO(n+1)$. Therefore, the zeros within $\mathbb{RP}^n$ are also invariant under the action of $SO(n+1)$.
\end{proof}

\begin{proposition} \label{InvarianceGAF}
The zeroes of the system $\bm{f}=\left(f_1, \dots, f_n\right)$ from (\ref{GAF}) are invariant under any isometry of $\mathbb{R}^n$.
That is, for any open set $U\subset \mathbb{R}^n$ and any isometry $I:
\mathbb{R}^n\rightarrow\mathbb{R}^n$, we have 
$\mbox{Pr}\left(\bm{f} \mbox{ has a zero in } U \right) = \mbox{Pr}\left(\bm{f}
\mbox{ has a zero in } I\left(U\right) \right)$.
\end{proposition} 

\begin{proof}
The zeroes of $\bm{f}$ are the same as those of 
\begin{equation} \label{G} \bm{g}:=\left(g_1,g_2,\dots,g_n\right) \quad \mbox{where} \quad g_i(\bm{x}):=\mathrm{e}^{-\frac{1}{2} ||\bm{x}||^2} f_i(\bm{x}).\end{equation} Each $g_i(\bm{x})$ defines a Gaussian process on $\mathbb{R}^n$, with mean 0 and covariance function 
\begin{equation} \label{Covariance} E\left(g_i(\bm{x})g_i(\bm{y})\right)={\rm e}^{-\frac{1}{2} \left( ||\bm{x}||^2 + ||\bm{y}||^2\right)} \sum\limits_{\bm{\alpha}} \frac{\bm{x}^{\bm{\alpha}} \bm{y}^{\bm{\alpha}}}{\bm{\alpha}!} 
={\rm e}^{-\frac{1}{2} \left( ||\bm{x}||^2 + ||\bm{y}||^2\right)}\prod\limits_{i=1}^n \left(\sum\limits_{\alpha_i = 0}^{\infty} \frac{\left(x_i y_i\right)^{\alpha_i}}{\alpha_i!} \right)
={\rm e}^{-\frac{1}{2}  ||\bm{x}-\bm{y}||^2}. \end{equation}
The result follows because (\ref{Covariance}) is clearly invariant under isometries of $\mathbb{R}^n$.
\end{proof}

We will now use these invariance  properties to reduce the proofs of Theorems \ref{CORR_FINITE_D} and \ref{CORR_ISOM_INV} to a particularly simple pairs of points 
and to local coordinates.
The two points 
\begin{eqnarray}\label{DEF_pq}
\bm{x} = \left[1:0:\cdots:0:-\frac{t}{2}\right] \quad \mbox{and} \quad \bm{y} = \left[1:0:\cdots:0:\frac{t}{2}\right]
\end{eqnarray}
(given here in homogeneous coordinates)
have distance 
\begin{eqnarray*}
\dist_{\mathbb{RP}^n}(\bm{x},\bm{y}) = 2 \, {\rm arctan}\left(\frac{t}{2}\right) = t + O(t^3) \quad \mbox{as $t \rightarrow 0$.}
\end{eqnarray*}
Thus, in order to prove Theorem 1, it suffices to verify (\ref{SHORT_RANGE_D}) for this pair of points.

Note that $(x_1,\ldots,x_n) \mapsto \left[1:x_1:\ldots:x_n\right]$ provides a
system of local coordinates in a neighborhood of $\bm{x}$ and~$\bm{y}$.  In
these coordinates, the $SO(n+1)$-invariant ensemble becomes
$$\bm{f}_d=\left(f_{d,1}(\bm{x}), f_{d,2}(\bm{x}), \dots,
f_{d,n}(\bm{x})\right),$$ where each $f_{d,i}$ is chosen independently of the
form
\begin{equation} \label{AffineEnsemble}
f_d(\bm{x}) = \sum\limits_{|\bm{\alpha}| \leq d} \sqrt{\binom{d}{\bm{\alpha}}} a_{\bm{\alpha}} \bm{x}^{\bm{\alpha}} \quad \mbox{where} \quad \binom{d}{\bm{\alpha}} = \frac{d!}{(d-|\bm{\alpha}|)!\prod\limits_{i=1}^{n} \alpha_{i}!}.
\end{equation}
and the $a_{\bm{\alpha}}$ are iid on the standard normal distribution $\mathcal{N}\left(0, 1\right)$. 

In summary:
Let $\K_{n,d}(\bm{x},\bm{y})$ and $\K_n(\bm{x},\bm{y})$ denote the correlation functions between zeros
of the $SO(n+1)$-invariant ensemble, expressed in affine coordinates (\ref{AffineEnsemble}), and between zeros of the $\isoR$-invariant ensemble (\ref{GAF}), respectively,
and let
\begin{align}
\K_{n,d}(t) := \K_{n,d}((0,\ldots,0,-t/2),(0,\ldots,0,t/2))\quad  \mbox{and} \quad \K_{n}(t) := \K_{n}((0,\ldots,0,-t/2),(0,\ldots,0,t/2))
\end{align}
 In order to prove
Theorems \ref{CORR_FINITE_D} and \ref{CORR_ISOM_INV}, it suffices to prove:
\begin{thm}\label{THM:LOCAL} 
{\color{white} jj}
\begin{enumerate}
\item 
As $t \rightarrow 0$ we have
\begin{eqnarray*}
\K_{n,d}(t) = A_{n,d} \ t^{2-n} + O(t^{3-n}), \quad \mbox{and} \quad \K_{n}(t) = A_{n} \ t^{2-n} + O(t^{3-n}), 
\end{eqnarray*}
where $A_{n,d}$ and $A_n$ are given in (\ref{SHORT_RANGE_D}) and (\ref{SHORT_RANGE}), respectively.
\item 
As $t \rightarrow \infty$ we have
\begin{eqnarray*}
\K_{n}(t) &=& 1 + O\left(t \mathrm{e}^{-\frac{t^2}{2}}\right).
\end{eqnarray*}
\end{enumerate}
\end{thm}

\section{Kac-Rice Formula} \label{Kac-Rice}

The main technique used in this paper is a variant of the classical
Kac-Rice Formula \cite{Kac,Rice1,Rice2} that  was developed for correlations
between zeros of multivariable Gaussian analytic functions by Bleher, Shiffman, and Zelditch in 
\cite[Section 2]{BSZUniversality}.    

We will begin this section with a statement and proof of the Kac-Rice formula
for the $m$ point correlation function with $m$ arbitrary.  At the end of the
section we will rephrase the results as needed in this paper.
Let $\bm{h}: \mathbb{R}^n \rightarrow \mathbb{R}^n$ be a Gaussian analytic
function and let $\bm{x}^1,\ldots,\bm{x}^m$ be $m$ distinct points in
$\mathbb{R}^n$.  
The {\em $m$-point correlation function} for the zeros of $\bm{h}$ is
\begin{eqnarray}\label{DEF_MULTI_CORR_FUNCTION}
K_n(\bm{x}^1,\ldots,\bm{x}^m) := \lim\limits_{\delta \rightarrow 0}  \frac{
{\rm Pr}\left(\exists\text{ a zero of } \bm{h} \text{ in } N_\delta({\bm{x}^i}) \mbox{ for each $i=1,\ldots,m$}\right)}
{\prod_{i=1}^m {\rm Pr}\left(\exists\text{ a zero of } \bm{h} \text{ in } N_\delta({\bm{x}^i})\right)},
\end{eqnarray}
where $N_\delta(\bm{x}^i)$ is the ball of radius $\delta > 0$ centered at
$\bm{x}^i \in \mathbb{R}^n$.

The  {\em $m$-point density} for the zeros of $\bm{h}$ is
\begin{eqnarray}\label{DEF_MULTI_DENSITY}
\KK_n(\bm{x}^1,\ldots,\bm{x}^m) := \lim\limits_{\delta \rightarrow 0} \frac{1}{(V_\delta)^m} 
{\rm Pr}\left(\exists\text{ a zero of } \bm{h} \text{ in } N_\delta({\bm{x}^i}) \mbox{ for each $i=1,\ldots,m$}\right),
\end{eqnarray}
where $V_\delta = \frac{\pi^{n/2}
\delta^n}{\Gamma\left(\frac{n}{2}+1\right)}$ is the volume of each ball.    When $m = 1$, 
$\KK_n(\bm{x}) = \rho_n(\bm{x})$, the probability density (\ref{DEF_DENSITY}).
For $m > 1$ we have
\begin{eqnarray}\label{EQN_MPOINT_DENSITY_VS_CORR}
K_n(\bm{x}^1,\ldots,\bm{x}^m) = \frac{1}{\prod_{i=1}^m \rho_n(\bm{x}^i)} \KK_n(\bm{x}^1,\ldots,\bm{x}^m).
\end{eqnarray}

Consider the Gaussian random $(m n^2 + mn)$-dimensional column vector
\begin{align}  \label{Vector_General}
\bm{v}:=  \begin {array}{ccccccccccc} \big[\nabla h_{1}(\bm{x}^1) & \nabla h_{2}(\bm{x}^1) & \ldots & \nabla h_{n}(\bm{x}^1) & \hdots & \hdots & \nabla h_{1}(\bm{x}^m) & \nabla h_{2}(\bm{x}^m) & \ldots & \nabla h_{n}(\bm{x}^m) \\
 & & & & & &  {\bm h}(\bm{x}^1) & {\bm h}(\bm{x}^2) & \cdots & \qquad \,\, {\bm h}(\bm{x}^m) \big]^\intercal,
\end{array}
\end{align}
where each gradient vector $\nabla h_{i}(\bm{x}^j)$ and each vector ${\bm h}(\bm{x}^k)$ is concatenated
into $\bm{v}$ in the indicated location.

Let $\bm \xi^1, \ldots, \bm \xi^m$ be $n\times n$ matrices whose rows are $\bm{\xi}^i_{1}, \ldots \bm{\xi}^i_{n}$
for each $i=1,\ldots,m$.  Let
\begin{align}\label{EQN_DEF_U}
\bm{u} =\left[ \begin {array}{cccccccccc} \bm \xi^1_{1} &  \ldots & \bm \xi^m_1 & \bm \xi^1_2 & \ldots & \bm \xi^m_2 & \ldots  & \ldots & \bm \xi^1_n & \ldots \bm \xi^m_n \end {array} \right]^{\intercal}
\end{align}
be the $m  n^2$ dimensional column vector formed by concatenating the first rows of
each of the matrices $\bm \xi^1, \ldots, \bm \xi^m$ followed by their second
rows, etc.

\begin{proposition}[\bf General Kac-Rice Formula]\label{KAC_RICE_MULTI_DENSITY}
Suppose the covariance matrix $\bold{C}=(Ev_i v_j)_{i,j=1}^{mn^2+mn}$ of the vector~(\ref{Vector_General}) is positive definite.
Then, the $m$-point density $\KK_n(\bm{x}^1,\ldots,\bm{x}^m)$ for the zeroes of the system $\bm{h}$ is given by
\begin{equation} \label{K1_general} 
\KK(\bm{x}^1,\ldots,\bm{x}^m)=\frac{1}{\left(2\pi\right)^{m n\left(n+1\right)/2} \sqrt{\det \bold{C}}} \int\limits_{\mathbb{R}^{m n^2}}  \prod_{i=1}^m  | \det \bm{\xi}^i| \mathrm{e}^{-\frac{1}{2} \left(\bold{\Omega}\bm{u},\bm{u}\right)} d\bm{u},
\end{equation}
where $\bold{\Omega}$ is the $mn^2 \times mn^2$ principal minor of $\bold{C}^{-1}$ and $\bm{u}$ is as in (\ref{EQN_DEF_U}).
\end{proposition}

Proposition \ref{KAC_RICE_MULTI_DENSITY} is easily obtained from \cite[Theorem
2.2]{BSZUniversality} by using the suitable Gaussian density $D_k(0,\xi,z)$ in
their formula (38).    In order for this paper to be relatively self contained, we present a  proof of Proposition \ref{KAC_RICE_MULTI_DENSITY} below.

We start with the following lemma. Define the derivative $D{\bm h}(\bm{x})$ as a linear map $\bm{w} \mapsto  D\bm{w}$, where $D$ is the matrix
\begin{align*}
D = \left(\frac{\partial h_i}{\partial x_j}({\bm x})\right)_{i,j=1}^n,
\end{align*}
and $\bm{w}=[w_1,\ldots,w_n]^\intercal$.

\begin{lemma}\label{LEM_EQUIV_VERSION_MULTI_DENSITY}  We have
\begin{eqnarray}\label{MULTI_DENSITY2}
\KK_n(\bm{x}^1,\ldots,\bm{x}^m) = \lim\limits_{\delta \rightarrow 0} \frac{1}{(V_\delta)^m}
{\rm Pr}\Big(\bm{h}(\bm{x}^i) \in {D{\bm h}(\bm{x}^i)}\big(N_\delta(\bm{0})\big) \mbox{ for each $i=1,\ldots,m$}\Big).
\end{eqnarray}
\end{lemma}

\begin{proof}
We begin by cutting off the tails of $\bm{h}$.
Let $R > 0$ be chosen sufficiently large so that for all $\delta > 0$
sufficiently small and all $i=1,\ldots,m$ we have $N_\delta(\bm{x}^i) \subset
\{\|\bm x\| \leq R\}$.  
Consider the following bounds on the derivatives of ${\bm h}$:
\begin{align}\label{EQN_DERIVATIVE_BOUNDS}
\|D{\bm h}({\bm x}^i)\| < A \mbox{ for all $1\leq i \leq m$ and } \left|\frac{\partial^2 h_j(\bm{x})}{\partial x_k \partial x_l}\right| < A \mbox{ for all $1 \leq j,k,l \leq n$ and all $\|\bm x\| \leq R$}.
\end{align}
For any Gaussian analytic function ${\bm h}$
\begin{align*}
\Pr\Big({\bm h} \mbox{ satisfies  condition (\ref{EQN_DERIVATIVE_BOUNDS}) } \Big)  \rightarrow 1 \mbox{ as } A \rightarrow \infty.
\end{align*}
Therefore, it will be sufficient to prove (\ref{MULTI_DENSITY2}) under the hypotheses 
(\ref{EQN_DERIVATIVE_BOUNDS}).  (The constant $A > 0$ will be fixed for the remainder of the proof.)

As in the statement of Proposition \ref{KAC_RICE_MULTI_DENSITY}, 
let $\bm{\xi}^i$ be the $n\times n$ matrices, for $i=1,\ldots,m$, and let $\bm{u} \in \mathbb{R}^{mn^2}$ be given as in (\ref{EQN_DEF_U}).
Let $\bm{s}~:=~\left[ \begin {array}{cccccccccc} \bm{s}^1 & \bm{s}^2 & \cdots &
\bm{s}^m \end{array} \right]^\intercal$ be the $mn$-dimensional column vector, where each $\bm{s}^i \in
\mathbb{R}^n$.
For any open subset $U \subset \mathbb{R}^n$ having compact closure and any $\epsilon > 0$ let
\begin{align*}
U_\epsilon^- := \{\bm{x} \in U \,: \, \dist(\bm{x},\partial U) \geq \epsilon\} \qquad \mbox{and} \qquad 
U_\epsilon^+ := \{\bm{x} \in \mathbb{R}^n \,: \, \dist(\bm{x},\overline U) \leq \epsilon\}.
\end{align*}
For any $\delta > 0$ and $B > 0$ consider 
\begin{align}
E^\pm_{\delta,B} &:= \left\{(\bm{u},\bm{s}) \in \mathbb{R}^{mn^2+mn} \,\, : \,\, \|\bm{\xi}^i\| < A \mbox{ and } \bm{s}^i \in (\bm{\xi}^i(N_\delta(\bm{0})))_{B \delta^2}^\pm  \, \mbox{for each $i=1,\ldots, m$}\right\} \qquad \mbox{and} \\
E^0_\delta &:= \left\{(\bm{u},\bm{s}) \in \mathbb{R}^{mn^2+mn} \,\, : \,\,  \|\bm{\xi}^i\| < A \mbox{ and } \bm{s}^i \in \bm{\xi}^i(N_\delta(\bm{0}))  \, \mbox{for each $i=1,\ldots, m$}\right\}.
\end{align} 
{
Here, 
$\bm{\xi}^i(N_\delta(\bm{0}))$ denotes the image of the ball $N_\delta(\bm{0})$ under the the linear map from $\mathbb{R}^n$ to $\mathbb{R}^n$ expressed in terms of the standard basis on $\mathbb{R}^n$ by the $n \times n$ matrix $\bm{\xi}^i$.
}

The bounds $\|\bm{\xi}^i\| < A$, for $i = 1,\ldots,m$, imply that as $\delta \rightarrow 0$ for fixed $B > 0$ we have
\begin{align}\label{EQN_VOLUMES}
{\rm Vol} (E^+_{\delta,B} \setminus E^-_{\delta,B}) = O(\delta^{mn+m}).
\end{align}

Let 
\begin{align*}
H_\delta &:= \left\{\bm{h} \,\,: \,\,  \bm{h} \text{ has a zero in } N_\delta({\bm{x}^i}) \mbox{ for each $i=1,\ldots,m$}\right\}, \\
H_{\delta,B}^{\pm} &:= \left\{\bm{h} \,\, : \,\, \left(D\bm{h}(\bm{x}^1),\ldots,D\bm{h}(\bm{x}^m),\bm{h}(\bm{x}^1),\ldots,\bm{h}(\bm{x}^m)\right) \in E^{\pm}_{\delta,B}\right\}, \quad \mbox{and} \\
H_\delta^{0} &:= \left\{\bm{h} \,\, : \,\, \left(D\bm{h}(\bm{x}^1),\ldots,D\bm{h}(\bm{x}^m),\bm{h}(\bm{x}^1),\ldots,\bm{h}(\bm{x}^m)\right) \in E^{0}_\delta \right\}.
\end{align*}
It is immediate from the definition of the sets $E^\pm_{\delta,B}$ and $E^0_\delta$ that for any $\delta > 0$
and any $B > 0$ that
\begin{align}\label{EQN_SUBSETS1}
H_{\delta,B}^- \subset H^0_\delta \subset H_{\delta,B}^+.
\end{align}
Meanwhile, by assumption (\ref{EQN_DERIVATIVE_BOUNDS}),
Taylor's Theorem gives that there exists $B > 0$ (depending on the bound $A$) so that for all sufficiently small $\delta > 0$ we have
\begin{align}\label{EQN_SUBSETS2}
H_{\delta,B}^- \subset H_\delta \subset H_{\delta,B}^+.
\end{align}

Since $\left(D\bm{h}(\bm{x}^1),\ldots,D\bm{h}(\bm{x}^m),\bm{h}(\bm{x}^1),\ldots,\bm{h}(\bm{x}^m)\right)$ is a Gaussian random vector, its probability distribution is absolutely continuous.  Therefore, (\ref{EQN_VOLUMES}) implies that
\begin{align}
\Pr\Big(\bm{h} \in H_{\delta,B}^+ \setminus H_{\delta,B}^-\Big) = O(\delta^{mn+m}).
\end{align}
Combined with (\ref{EQN_SUBSETS1}) and (\ref{EQN_SUBSETS2}) this implies that under the assumption (\ref{EQN_DERIVATIVE_BOUNDS}) we have
\begin{align} 
\lim\limits_{\delta \rightarrow 0} \frac{1}{(V_\delta)^m} \Pr\Big(\bm h \in H_\delta\Big) = \lim\limits_{\delta \rightarrow 0} \frac{1}{(V_\delta)^m} \Pr\Big(\bm h \in H^0_\delta\Big),
\end{align}
since $(V_\delta)^m$ is bounded below by a constant times  $\delta^{mn}$.

\end{proof}

\noindent
{\bf Proof of Proposition \ref{KAC_RICE_MULTI_DENSITY}.}
We use Lemma \ref{LEM_EQUIV_VERSION_MULTI_DENSITY} to replace the definition of $\KK_n$ with (\ref{MULTI_DENSITY2}).
Using the formula for a Gaussian density, we have
\begin{align}
\KK_n(\bm{x}^1,\ldots,\bm{x}^m) =  \frac{1}{\left(2\pi\right)^{m
n\left(n+1\right)/2} \sqrt{\det \bold{C}}} \lim\limits_{\delta \rightarrow 0}
\frac{1}{(V_\delta)^m} \int\limits_{\mathbb{R}^{m n^2}} \,\,
\int\limits_{\bm{\xi}^1(N_\delta(\bm{0})) \times \cdots \times \bm{\xi}^m(N_\delta(\bm{0})) }
 \mathrm{e}^{-\frac{1}2
\left(\bold{C^{-1}}\left[\begin{array}{c} \bm{u} \\ \bm{s}\end{array}\right],\left[\begin{array}{c} \bm{u} \\ \bm{s}\end{array}\right] \right)} d\bm{s} d\bm{u},
\end{align}
where $\left[\begin{array}{c} \bm{u} \\ \bm{s}\end{array}\right]$ denotes
the column vector obtained by stacking the two column vectors $\bm{u}$ and
$\bm{s}$.

Because $\bold{C}$ is positive definite, the integrand decays rapidly at infinity.  Thus, 
the Dominated Convergence Theorem allows us to interchange the limit with the first integral.  We also multiply
and divide by $\prod_{i=1}^m  | \det \bm{\xi}^i|$, obtaining:
\begin{align}
& \KK_n(\bm{x}^1,\ldots,\bm{x}^m) =  \\ 
 & \frac{1}{\left(2\pi\right)^{m n\left(n+1\right)/2} \sqrt{\det \bold{C}}} 
\int\limits_{\mathbb{R}^{m n^2}}  \lim\limits_{\delta \rightarrow 0}
\frac{\prod_{i=1}^m  | \det \bm{\xi}^i|}{\prod_{i=1}^m  {\rm Vol} \, \bm{\xi}^i(N_\delta(\bm{0}))} 
\int\limits_{\bm{\xi}^1(N_\delta(\bm{0})) \times \cdots \times
\bm{\xi}^m(N_\delta(\bm{0})) } \mathrm{e}^{-\frac{1}2\left(\bold{C^{-1}}\left[\begin{array}{c} \bm{u} \\ \bm{s}\end{array}\right],\left[\begin{array}{c} \bm{u} \\ \bm{s}\end{array}\right] \right)} d\bm{s} d\bm{u}. \nonumber
\end{align}
The Integral Mean Value Theorem implies that
\begin{align}
 \lim\limits_{\delta \rightarrow 0} \frac{1}{\prod_{i=1}^m  {\rm
Vol} \, \bm{\xi}^i(N_\delta(\bm{0}))}  \int\limits_{\bm{\xi}^1(N_\delta(\bm{0}))
\times \cdots \times \bm{\xi}^m(N_\delta(\bm{0})) }
\mathrm{e}^{-\frac{1}2\left(\bold{C^{-1}}\left[\begin{array}{c} \bm{u} \\
\bm{s}\end{array}\right],\left[\begin{array}{c} \bm{u} \\
\bm{s}\end{array}\right] \right)} d{\bm s} 
= \mathrm{e}^{-\frac{1}2\left(\bold{C^{-1}}\left[\begin{array}{c} \bm{u} \\
\bm{0}\end{array}\right],\left[\begin{array}{c} \bm{u} \\
\bm{0}\end{array}\right] \right)} =  \mathrm{e}^{-\frac{1}{2} \left(\bold{\Omega}\bm{u},\bm{u}\right)},
\end{align}
which completes the proof.
\qed

\begin{remark}\label{RMK_NBHDS}
Proposition \ref{KAC_RICE_MULTI_DENSITY} shows that the correlation measure is
absolutely continuous off of the ``diagonal'' where $\bm{x}^i = \bm{x}^j$ for some $i
\neq j$ (hence the name ``correlation function'').  Thus, in the definition
(\ref{DEF_MULTI_CORR_FUNCTION}) of $K_n(\bm{x}^1,\ldots,\bm{x}^m)$ one need not
use round balls $N_\delta(\bm{x}^i)$.  Rather, any sequence of shrinking neighborhoods
of each $\bm{x}^i$ suitable for computing a Radon-Nikodym derivative will
suffice.
\end{remark}

On certain occasions we will need the following lemma, which is proved in Appendix \ref{SEC:Lemma8}, to make estimates involving the Kac-Rice formula (\ref{K1_general}).
\begin{lemma} \label{COCV3}
For any positive definite $mn^2 \times mn^2$ matrix $\bold{A}$
\begin{equation}
\left| \, \int\limits_{\mathbb{R}^{mn^2}}  \prod_{i=1}^m  | \det \bm{\xi}^i| \mathrm{e}^{-\frac{1}{2} \left(\bold{B}\bm{u},\bm{u}\right)}  -  \int\limits_{\mathbb{R}^{mn^2}}  \prod_{i=1}^m  | \det \bm{\xi}^i|  \mathrm{e}^{-\frac{1}{2} \left(\bold{A}\bm{u},\bm{u}\right)} d\bm{u} \right| =  O\left(||\bold{A}-\bold{B}||_{\infty}^{1/2}\right)
\end{equation}
for any $mn^2\times mn^2$ matrix $\bold{B}$ sufficiently close to $\bold{A}$. (Here $\bm{u}$ is as in (\ref{EQN_DEF_U}) and $||\ ||_{\infty}$ denotes the maximum entry of the matrix.)
\end{lemma}

{
\begin{remark}
After reading a preprint of this paper, L. Nicolaescu informed us that a very similar estimate appears in
his paper [Prop. A.1]\cite{Nicolaescu}, whose proof was provided by G. Lowther in a discussion on Math Overflow \cite{MO}.
\end{remark}
}

We close the section with simplified rephrasings of Proposition
\ref{KAC_RICE_MULTI_DENSITY} in the cases $m=1$ and $m=2$ that will be used
throughout this paper.
In both cases it will be better to order the random vector $\bm{v}$ from (\ref{Vector_General}) in a way that will make the covariance matrix $\bold{C}$ block diagonal.
(It will come at the cost of the definition for $\bold{\Omega}$ being slightly more complicated.)

If $m=1$ we reorder the random vector $\bm{v}$ as
\begin{equation}  \label{Vector_density}
\bm{v}:= \left[ \begin {array}{ccccccccc}{ h_1(\bm{x})}&{ \nabla h_{1}(\bm{x})}&\dots&h_n(\bm{x})&\nabla h_{n}(\bm{x})\end {array} \right]^\intercal.
\end{equation}
Because the components of $\bm{h}$ are chosen iid, the resulting covariance matrix
$\bold{C}=(Ev_i v_j)_{i,j=1}^{n^2+n}$ will be of the form $\bold{C} = \diag_n(\tilde{\bold{C}})$, where
$\tilde{\bold{C}} = (Ev_i v_j)_{i,j=1}^{n+1}$ and $\diag_n\left(\tilde{\bold{C}}\right)$ denotes the block diagonal matrix with $n$ copies of $\tilde{\bold{C}}$ on the diagonal.  
The vector $\bm{u}$ becomes
$\begin{displaystyle} \bm{u} = \left[ \begin {array}{ccc} { {\bm \xi}_{1}}&\dots&{
{\bm \xi}_{n}}\end {array} \right]^{\intercal}
\end{displaystyle}$ where $\bm{\xi}$ is a $n\times n$ matrix.
\begin{proposition}{\bf (Kac-Rice for Density)} \label{KacRiceFormula_density} Suppose the
covariance matrix $\bold{C}=(Ev_i v_j)_{i,j=1}^{n(n+1)}$ of the vector (\ref{Vector_density}) is positive definite.
Then, the density of zeroes of the system $\bm{h}$ is:
\begin{equation}\label{K_density} \rho_n(\bm{x})=\frac{1}{\left(2\pi\right)^{n\left(n+1\right)/2}\sqrt{\det \bold{C}}} \int\limits_{\mathbb{R}^{n^2}} | \det {\bm \xi}| \mathrm{e}^{-\frac{1}{2} \left(\bold{\Omega}{\bm u},{\bm u}\right)} d{\bm u},\end{equation}
where $\bold{\Omega}$ is the matrix of the elements of $\bold{C}^{-1}$ left after removing the rows and columns with indices congruent to 1 modulo $n+1$.
\end{proposition}

If $m=2$ we reorder $\bm{v}$ as
\begin{align}  \label{Vector}
\bm{v}:= \left[ \begin {array}{ccccccccccccc}{ h_1(\bm{x})}&{ \nabla h_{1}(\bm{x})}&{ h_1(\bm{y})}&{ \nabla h_{1}(\bm{y})}&\dots&h_n(\bm{x})&\nabla h_{n}(\bm{x})&h_n(\bm{y})&\nabla h_{n}(\bm{y})\end {array} \right]^\intercal,
\end{align}
which will again make its covariance matrix block diagonal $\bold{C} = \diag_n(\tilde{\bold{C}})$.
Let $\begin{displaystyle}{\bm \xi}\end{displaystyle}$ and $\bm{\eta}$ be the
$n\times n$ matrices whose rows are ${\bm \xi}_{1}, \dots {\bm \xi}_{n}$ and
${\bm \eta}_{1}, \dots {\bm \eta}_{n}$, respectively. Let
$\begin{displaystyle}\bm{u} =\left[ \begin {array}{ccccccc} { \bm{\xi}_{1}}&{
\bm{\eta}_{1}}&{ \bm{\xi}_{2}}&{ \bm{\eta}_{2}}&\dots&{ \bm{\xi}_{n}}&{
\bm{\eta}_{n}}\end {array} \right]^{\intercal} \end{displaystyle}$ be the vector
formed by alternating the vectors $\bm{\xi}_{i}$ and~$\bm{\eta}_{i}$.

\begin{proposition}[\bf Two Point Kac-Rice Formula] \label{KacRiceFormula}
Suppose the covariance matrix $\bold{C}=(Ev_i v_j)_{i,j=1}^{2n(n+1)}$ of the vector~(\ref{Vector}) is positive definite.
Then, the two-point correlation function for the zeroes of the system $\bm{h}$ is:
\begin{equation} \label{K1} K_n\left(\bm{x},\bm{y}\right)=\frac{1}{\left(2\pi\right)^{n\left(n+1\right)} \rho(\bm{x})\rho(\bm{y})\sqrt{\det \bold{C}}} \int\limits_{\mathbb{R}^{2n^2}} | \det \bm{\xi}| |\det \bm{\eta}| \mathrm{e}^{-\frac{1}{2} \left(\bm{\Omega}\bm{u},\bm{u}\right)} d\bm{u},\end{equation}
where $\bold{\Omega}$ is the matrix of the elements of $\bold{C}^{-1}$ left after removing the rows and columns with indices congruent to 1 modulo $n+1$.
\end{proposition}

\section{Calculation of the covariance matrices, their inverses, and $\bm{\Omega}$} \label{SEC:COVARIANCE}

Let $\bold{C}_{n,d} \equiv \bold{C}_{n,d}(t)$ and $\bold{C}_n \equiv
\bold{C}_n(t)$ be the covariance matrix for vector (\ref{Vector}) applied to
$\bm{f}_d$ (Equation \ref{AffineEnsemble}) and $\bm{f}$ (Equation \ref{GAF}), respectively, at the points 
\begin{align}\label{EQN:DEF_XY}
\bm{x} = \left(0,\ldots,0,-\frac{t}{2}\right) \qquad \mbox{and} \qquad  \bm{y} = \left(0,\ldots,0,\frac{t}{2}\right).
\end{align}

\begin{lemma} \label{COV_MATRIX_CALC}
Both $\bold{C}_{n,d}$ and $\bold{C}_n$ are of the form 
\begin{align}\label{EQN_GENERIC_C}
\bold{C} := \diag_n(\tilde{\bold{C}}), \qquad  \mbox{with} \qquad 
\bold{\tilde{C}}= \left[ \begin{matrix} \bold{A}_+ & \bold{B}^{\intercal} \\
\bold{B}  &  \bold{A}_- \\ \end{matrix} \right],
\end{align}
 where $\bold{A}_\pm$ and $\bold{B}$ are the following $(n+1) \times (n+1)$ matrices:
 \begin{align}\label{EQN_GENERIC_C_DETAILS}  \bold{A}_\pm = \left[ \begin{matrix}
\alpha &  0  & \dots & 0 & \pm \delta \\
0  &  \beta & \ddots &  & 0 \\
\vdots & \ddots & \ddots & \ddots & \vdots \\
0 &   & \ddots & \beta & 0\\
\pm \delta & 0 & \dots & 0 & \gamma \end{matrix}\right],  \quad
 \bold{B} = \left[ \begin{matrix}
\mu &  0  & \dots & 0 & \nu \\
0  &  \eta & \ddots &  & 0 \\
\vdots & \ddots & \ddots & \ddots & \vdots \\
0 &   & \ddots & \eta & 0\\
-\nu & 0 & \dots & 0 & \tau \end{matrix}\right]
\end{align}
and where $\alpha, \beta, \delta, \gamma, \mu, \eta, \nu,$ and  $\tau$ are functions of $d$ and $t$ expressed in (\ref{ALPHA_SO}-\ref{TAU_SO}) for $\bold{C}_{n,d}$ and the functions of $t$ expressed in  (\ref{ALPHA_ISO}-\ref{TAU_ISO}) for $\bold{C}_n$.  
\end{lemma}

\begin{proof}

Since the coefficients of $f_{d,i}$ and $f_{d,j}$ (respectively $f_i$ and
$f_j$) are independent when $i\ne j$, only the entries of $\bold{C}_{n,d}$
(respectively $\bold{C}_n$) with $i = j$ will have nonzero values.
Thus, the covariance matrices will have the following
block-diagonal structure:
 \begin{equation}
 \bold{C}_{n,d} = {\rm diag}_n (\bold{\tilde{C}_{n,d}}) \quad \mbox{and} \quad \bold{C}_{n} = {\rm diag}_n (\bold{\tilde{C}_{n}}),
 \end{equation}
 where $\bold{\tilde{C}_{n,d}}$ corresponds to the first $2n+2$ entries of $\bm{v}$ (and similarly for $\bold{\tilde{C}_{n}}$).
These entries correspond to $f_{d,1}$ and $f_1$, respectively. For ease of notation, we'll drop the subscript $1$: $f_d \equiv f_{d,1}$, $f \equiv f_1$.

For any $\bm{z},\bm{w} \in \mathbb{R}^n$  we have:
\begin{align}
E(f_d(\bm{z}) f_d(\bm{w})) &= \left(1+\bm{z} \cdot \bm{w} \right)^d, \label{SO_FIRST}\\ 
E\left(f_d(\bm{z}) \frac{\partial f_d(\bm{w})}{\partial w_i} \right) &=  \frac{\partial E(f_d(\bm{z}) f_d(\bm{w}))}{\partial w_i} = d \ z_i  \left(1+\bm{z} \cdot \bm{w} \right)^{d-1},  \\
E\left(\frac{\partial f_d(\bm{z})}{\partial z_i} \frac{\partial
f_d(\bm{w})}{\partial w_i} \right) &= \frac{\partial^2 E(f_d(\bm{z})
f_d(\bm{w}))}{\partial z_i \partial w_i}  = d \left(1+\bm{z} \cdot \bm{w}
\right)^{d-1} +  d(d-1) z_i w_i \left(1+\bm{z} \cdot \bm{w} \right)^{d-2}, \quad \mbox{and} \\
E\left(\frac{\partial f_d(\bm{z})}{\partial z_i} \frac{\partial f_d(\bm{w})}{\partial w_j} \right) &= \frac{\partial^2 E(f_d(\bm{z}) f_d(\bm{w}))}{\partial z_i \partial w_j} = d(d-1) z_j w_i  \left(1+\bm{z} \cdot \bm{w} \right)^{d-2} \quad \mbox{for} \quad i \neq j. \label{SO_LAST}
\end{align}
Recalling that $\bm{x} = \left(0,\ldots,0,-\frac{t}{2}\right)$ and $\bm{y} = \left(0,\ldots,0,\frac{t}{2}\right)$ we now use expressions (\ref{SO_FIRST}-\ref{SO_LAST}) to compute that 
the only non-zero covariances in $\bold{\tilde{C}_{n,d}}$ are
\begin{align}
\alpha &:= E(f_d(\bm{x}) f_d(\bm{x})) = E(f_d(\bm{y}) f_d(\bm{y})) = \left( 1+\frac{{t}^{2}}{4} \right) ^{d}, \label{ALPHA_SO} \\
\delta &:= E\left(f_d(\bm{x}) \frac{\partial f_d(\bm{x})}{\partial x_n} \right) -E\left(f_d(\bm{y})  \frac{\partial f_d(\bm{y})}{\partial y_n} \right) = -d \frac{t}{2}  \left(1+\frac{{t}^{2}}{4}\right)^{d-1},\\ 
\beta &:= E\left(\frac{\partial f_d(\bm{y})}{\partial y_i} \frac{\partial
f_d(\bm{y})}{\partial y_i} \right)  = E\left(\frac{\partial f_d(\bm{y})}{\partial y_i} \frac{\partial
f_d(\bm{y})}{\partial y_i} \right) = d \left( 1+\frac{{t}^{2}}{4} \right)
^{d-1} \quad \mbox{for $i \neq n$}, \\
\gamma &:= E\left(\frac{\partial f_d(\bm{x})}{\partial x_n} \frac{\partial f_d(\bm{x})}{\partial x_n} \right)  = E\left(\frac{\partial f_d(\bm{y})}{\partial y_n} \frac{\partial f_d(\bm{y})}{\partial y_n} \right)   = d \left( 1+\frac{{t}^{2}}{4} \right) ^{d-1}+\frac{1}{4}d \left( d-1 \right) {t}^{
2} \left( 1+\frac{{t}^{2}}{4} \right) ^{d-2}, \\
\mu &:= E(f_d(\bm{x}) f_d(\bm{y})) = \left(1+\bm{x} \cdot \bm{y} \right)^d = \left( 1-\frac{{t}^{2}}{4} \right) ^{d}, \\
\nu &:= E\left(f_d(\bm{x}) \frac{\partial f_d(\bm{y})}{\partial y_n} \right) =- E\left(f_d(\bm{y}) \frac{\partial f_d(\bm{x})}{\partial x_n} \right)  = - d \frac{t}{2}  \left(1 -\frac{{t}^{2}}{4}  \right)^{d-1},\\
\eta &:= E\left(\frac{\partial f_d(\bm{x})}{\partial x_i} \frac{\partial f_d(\bm{y})}{\partial y_i} \right) = d \left(1-\frac{t^2}{4} \right)^{d-1} \quad \mbox{for $i \neq n$}, \quad \mbox{and}\\
\tau &:=  E\left(\frac{\partial f_d(\bm{x})}{\partial x_n} \frac{\partial
f_d(\bm{y})}{\partial y_n} \right) = 
d \left(1-\frac{{t}^{2}}{4} \right) ^{d-1}-\frac{1}{4}d \left( d-1 \right) {t}^ {2}
\left(1-\frac{{t}^{2}}{4} \right) ^{d-2}. \label{TAU_SO} 
\end{align}
This proves that $\bold{C_{n,d}}$ has the structure stated in Lemma \ref{COV_MATRIX_CALC}.

For $\isoR$-invariant ensemble $\bm{f}$ and any $\bm{z}, \bm{w} \in \mathbb{R}^n$ we have:
\begin{eqnarray}
E(f(\bm{z}) f(\bm{w})) &=& e^{\bm{z} \cdot \bm{w}}, \label{ISO_FIRST} \\
E\left(f(\bm{z}) \frac{\partial f(\bm{w})}{\partial w_i} \right) &=& \frac{\partial E(f(\bm{z}) f(\bm{w}))}{\partial w_i} = z_i e^{\bm{z} \cdot \bm{w}}, \\
E\left(\frac{\partial f(\bm{z})}{\partial z_i} \frac{\partial f(\bm{w})}{\partial w_i} \right) &=& \frac{\partial^2 E(f(\bm{z}) f(\bm{w}))}{\partial z_i \partial w_i} = (1+z_i w_i) e^{\bm{z} \cdot \bm{w}}, \mbox{and} \\
E\left(\frac{\partial f(\bm{z})}{\partial z_i} \frac{\partial f(\bm{w})}{\partial w_j} \right) &=& \frac{\partial^2 E(f(\bm{z}) f(\bm{w}))}{\partial z_i \partial w_j} = z_j w_i e^{\bm{z} \cdot \bm{w}} \quad 
\mbox{for $i \neq j$ } \label{ISO_LAST}.
\end{eqnarray}
We now use (\ref{ISO_FIRST}-\ref{ISO_LAST}) to compute that
the only non-zero covariances in $\bold{\tilde{C}_{n}}$ are
\begin{eqnarray}
\alpha &:=&  E(f(\bm{x}) f(\bm{x})) = E(f(\bm{y}) f(\bm{y})) = e^{\frac{t^2}{4}}, \label{ALPHA_ISO} \\
\delta &:=& E\left(f(\bm{x}) \frac{\partial f(\bm{x})}{\partial x_n} \right) = -E\left(f(\bm{y}) \frac{\partial f(\bm{y})}{\partial y_n} \right)  = -\frac{t}{2}  e^{\frac{t^2}{4}}, \\
\beta &:=&  E\left(\frac{\partial f(\bm{y})}{\partial y_i} \frac{\partial f(\bm{y})}{\partial y_i} \right) = E\left(\frac{\partial f(\bm{y})}{\partial y_i} \frac{\partial f(\bm{y})}{\partial y_i} \right) =  e^{\frac{t^2}{4}} \quad \mbox{for $i \neq n$}, \\
\gamma &:=& E\left(\frac{\partial f(\bm{x})}{\partial x_n} \frac{\partial f(\bm{x})}{\partial x_n} \right) = E\left(\frac{\partial f(\bm{y})}{\partial y_n} \frac{\partial f(\bm{y})}{\partial y_n}\right) = \left(1+\frac{t^2}{4}\right)e^{\frac{t^2}{4}}, \\
\mu &:=& E(f(\bm{x}) f(\bm{y})) = e^{-\frac{t^2}{4}}, \\
\nu &:=& E\left(f(\bm{x}) \frac{\partial f(\bm{y})}{\partial y_n} \right) = - E\left(f(\bm{y}) \frac{\partial f(\bm{x})}{\partial x_n}\right) =  -\frac{t}{2}  e^{-\frac{t^2}{4}}, \\
\eta &:=& E\left(\frac{\partial f(\bm{x})}{\partial x_i} \frac{\partial f(\bm{y})}{\partial y_i} \right) = e^{-\frac{t^2}{4}} \quad \mbox{for $i \neq n$}, \quad  \mbox{and} \\
\tau &:=& E\left(\frac{\partial f(\bm{x})}{\partial x_n} \frac{\partial f(\bm{y})}{\partial y_n} \right) = \left(1-\frac{t^2}{4}\right) e^{-\frac{t^2}{4}}. \label{TAU_ISO}
\end{eqnarray}
This proves that $\bold{C}_{n}$ also has the structure stated in Lemma \ref{COV_MATRIX_CALC}.
\end{proof}

To apply the Kac-Rice formula to compute $K_{n,d}(\bm{x},\bm{y})$ and $K_n(\bm{x},\bm{y})$ for the values of
$\bm{x}$ and $\bm{y}$ given in (\ref{EQN:DEF_XY}) we will need to compute $\det
\bold{C_{n,d}}$, $\det\bold{C_n}$, $\bold{\Omega_{n,d}}$, $\bold{\Omega_{n}}$
and the diagonalizations of  $\bold{\Omega_{n,d}}$ and $\bold{\Omega_{n}}$.  
Here, $\bold{\Omega_{n,d}}$ and $\bold{\Omega_{n}}$ are the matrices obtained
from $\bold{C}_{n,d}^{-1}$ and $\bold{C}_n^{-1}$, respectively, by deleting all
of the rows and columns whose indices are congruent  1 modulo $n+1$, as in the
Kac-Rice formula.

By Lemma \ref{COV_MATRIX_CALC} we can do all of these calculations in terms of
the generic form $\bold{C}$ given in (\ref{EQN_GENERIC_C}) and then substitute
in the values of $\alpha$ through $\tau$ from (\ref{ALPHA_SO}-\ref{TAU_SO}) and
(\ref{ALPHA_ISO}-\ref{TAU_ISO}) accordingly.

\vspace{0.1in}

The determinant of $\bold{C}$ is
\begin{align}\label{DETC_GENERAL}
\det(\bold{C}) = \left( \beta^2-\eta^2 \right)^{n(n-1)} \left(
\alpha\,\gamma-\alpha\,\tau-{\delta}^{2}-2\,\delta\,\nu+\gamma\,\mu-
\mu\,\tau-{\nu}^{2} \right)^n  \left( \alpha\,\gamma+\alpha\,\tau-{
\delta}^{2}+2\,\delta\,\nu-\gamma\,\mu-\mu\,\tau-{\nu}^{2} \right)^n.
\end{align}
Recall that $\bold{C} = \diag_n(\bold{\tilde{C}})$ where $\bold{\tilde{C}}$ is described in (\ref{EQN_GENERIC_C}) and (\ref{EQN_GENERIC_C_DETAILS}).
Applying a suitable permutation to the rows and columns of $\tilde{\bold{C}}$, one obtains a block matrix 
with one $4 \times 4$ block and $n-1$ copies of the same $2\times 2$ block.  
Because of this, $\tilde{\bold{C}}^{-1}$ will have the same block structure and it can readily be computed
to be
$\begin{displaystyle}\bold{\tilde{C}}^{-1}= \left[ \begin{matrix}
\bold{D}_+ & \bold{E}_+ \\
\bold{E}_-  &  \bold{D}_- \\ \end{matrix} \right]
\end{displaystyle}$
 where $\bold{D}_\pm$ and $\bold{E}_\pm$ are the following $(n+1) \times (n+1)$ matrices:
 \begin{align*} \bold{D}_\pm &= 
\left[ \begin {array}{ccccc} {\frac {\alpha\,{\gamma}^{2}-\alpha\,{
\tau}^{2}-{\delta}^{2}\gamma-2\,\delta\,\nu\,\tau-\gamma\,{\nu}^{2}}{\Delta}}
&0&\dots&0& \mp {\frac {\alpha\,\delta\,\gamma+
\alpha\,\nu\,\tau-{\delta}^{3}-\delta\,\mu\,\tau+\delta\,{\nu}^{2}-
\gamma\,\mu\,\nu}{\Delta}} \\ \noalign{\medskip}0&{
\frac {\beta}{{\beta}^{2}-{\eta}^{2}}}&\ddots&0&0\\ \noalign{\medskip} \vdots & \ddots & \ddots
& \ddots & \vdots \\ \noalign{\medskip}0& &\ddots &{
\frac {\beta}{{\beta}^{2}-{\eta}^{2}}}&0\\  \noalign{\medskip} \mp {\frac {
\alpha\,\delta\,\gamma+\alpha\,\nu\,\tau-{\delta}^{3}-\delta\,\mu\,
\tau+\delta\,{\nu}^{2}-\gamma\,\mu\,\nu}{\Delta}} &0
& \dots &0&{\frac {{\alpha}^{2}\gamma-\alpha\,{\delta}^{2}-\alpha\,{\nu}^{2}
+2\,\delta\,\mu\,\nu-\gamma\,{\mu}^{2}}{\Delta}}
\end {array} \right] \\
 \bold{E}_\pm &=
\left[ \begin {array}{ccccc}  -{\frac {-{\delta}^{2}\tau-2\,\delta\,
\gamma\,\nu+{\gamma}^{2}\mu-\mu\,{\tau}^{2}-{\nu}^{2}\tau}{\Delta}} &0&\dots&0&  \mp{\frac {\alpha\,\delta\,\tau+\alpha\,\gamma\,
\nu+{\delta}^{2}\nu-\delta\,\gamma\,\mu-\mu\,\nu\,\tau-{\nu}^{3}}{\Delta}}
\\ \noalign{\medskip}0& -{\frac {\eta}{{
\beta}^{2}-{\eta}^{2}}} &\ddots&0&0\\ \noalign{\medskip} \vdots & \ddots & \ddots
& \ddots & \vdots \\ \noalign{\medskip}0& &\ddots &
-{\frac {\eta}{{\beta}^{2}-{\eta}^{2}}} &0\\  
\noalign{\medskip} \pm {\frac {\alpha\,\delta\,
\tau+\alpha\,\gamma\,\nu+{\delta}^{2}\nu-\delta\,\gamma\,\mu-\mu\,\nu
\,\tau-{\nu}^{3}}{\Delta}} &0
& \dots &0&
-{\frac {{\alpha}^{
2}\tau+2\,\alpha\,\delta\,\nu-{\delta}^{2}\mu-{\mu}^{2}\tau-\mu\,{\nu}
^{2}}{\Delta}}
\end {array} \right] 
\end{align*}

where
{\small
\begin{align}
\Delta = { \alpha}^{2}{\gamma}^{2}-{\alpha}^{2}{\tau}^{2}-2\,\alpha\,{\delta}^{2} \gamma-4\,\alpha\,\delta\,\nu\,\tau-2\,\alpha\,\gamma\,{\nu}^{2}+{ \delta}^{4}+2\,{\delta}^{2}\mu\,\tau-2\,{\delta}^{2}{\nu}^{2}+4\, \delta\,\gamma\,\mu\,\nu-{\gamma}^{2}{\mu}^{2}+{\mu}^{2}{\tau}^{2}+2\, \mu\,{\nu}^{2}\tau+{\nu}^{4}.
\end{align}
}

If $\bm{\Omega}$ is obtained from $\bold{C}^{-1}$
by deleting all of the rows and columns whose indices are congruent  1 modulo $n+1$, as in the Kac-Rice formula, we have

\begin{lemma} \label{OmegaLem}
$\begin{displaystyle}\bm{\Omega}=\diag_n\left(\bm{\tilde{\Omega}}\right)\end{displaystyle}$ with $\begin{displaystyle}\bm{\tilde{\Omega}}= \left[ \begin{matrix}
\bm{\tilde{\Omega}}_{\bold{1,1}} & \bm{\tilde{\Omega}_{1,2}} \\
\bm{\tilde{\Omega}_{2,1}}  & \bm{\tilde{\Omega}_{2,2}} \\ \end{matrix} \right]\end{displaystyle}$,
where:
\begin{align} 
\bm{\tilde{\Omega}_{1,1}}&=\bm{\tilde{\Omega}_{2,2}}=\diag \left(\underbrace{\frac {\beta}{{\beta}^{2}-{\eta}^{2}}, \dots,\frac {\beta}{{\beta}^{2}-{\eta}^{2}}}_{n-1 \ \rm{ times}},
{\frac {{\alpha}^{2}\gamma-\alpha\,{\delta}^{2}-\alpha\,{\nu}^{2}
+2\,\delta\,\mu\,\nu-\gamma\,{\mu}^{2}}{\Delta}}\right), \\ 
\bm{\tilde{\Omega}_{1,2}}&=\bm{\tilde{\Omega}_{2,1}}=\diag\left(\underbrace{-{\frac {\eta}{{\beta}^{2}-{\eta}^{2}}},\dots,-{\frac {\eta}{{\beta}^{2}-{\eta}^{2}}}}_{n-1 \ \rm{times}}, -{\frac {{\alpha}^{
2}\tau+2\,\alpha\,\delta\,\nu-{\delta}^{2}\mu-{\mu}^{2}\tau-\mu\,{\nu}
^{2}}{\Delta}} \right).
\end{align}
 \end{lemma}

We notice that there exists a permutation matrix $\bold{Q}$ such that {\begin{equation} \label{Omega} \bold{M}:= \bold{Q^\intercal }\bm{\Omega}\bold{Q}=\diag\left(\underbrace{\underbrace{\bold{M}_1,\dots,\bold{M}_1}_{n-1 \ {\rm times}},\bold{M}_2,\dots,\underbrace{\bold{M}_1,\dots,\bold{M}_1}_{n-1 \ {\rm times}},\bold{M}_2}\limits_{n \ {\rm times}} \right).\end{equation}}
\begin{displaymath} \mbox{where} \quad \bold{M}_1=\left[\begin{matrix}
\frac {\beta}{{\beta}^{2}-{\eta}^{2}} & -{\frac {\eta}{{\beta}^{2}-{\eta}^{2}}}  \\
-{\frac {\eta}{{\beta}^{2}-{\eta}^{2}}} & \frac {\beta}{{\beta}^{2}-{\eta}^{2}} \\ \end{matrix}  \right]  \quad \mbox{and} \quad \bold{M}_2=\left[\begin{matrix}
{\frac {{\alpha}^{2}\gamma-\alpha\,{\delta}^{2}-\alpha\,{\nu}^{2}
+2\,\delta\,\mu\,\nu-\gamma\,{\mu}^{2}}{\Delta}} &   -{\frac {{\alpha}^{
2}\tau+2\,\alpha\,\delta\,\nu-{\delta}^{2}\mu-{\mu}^{2}\tau-\mu\,{\nu}
^{2}}{\Delta}}  \\
 -{\frac {{\alpha}^{
2}\tau+2\,\alpha\,\delta\,\nu-{\delta}^{2}\mu-{\mu}^{2}\tau-\mu\,{\nu}
^{2}}{\Delta}} & 
{\frac {{\alpha}^{2}\gamma-\alpha\,{\delta}^{2}-\alpha\,{\nu}^{2}
+2\,\delta\,\mu\,\nu-\gamma\,{\mu}^{2}}{\Delta}} \\ \end{matrix}\right].
\end{displaymath}

\begin{lemma} \label{DiagonalizingMatrix}
We can orthogonally diagonalize $\bold{\Omega}$ with $\bold{Q} \bold{P}$, where $\bold{Q}$ is the permutation matrix described in the previous paragraph and 
$\begin{displaystyle}\label{DiagonalizingMatrixEQN} \bold{P} = \diag_{n^2}\left(\left[\begin{matrix} -\frac{\sqrt{2}}{2} & \frac{\sqrt{2}}{2} \\ \frac{\sqrt{2}}{2} & \frac{\sqrt{2}}{2} \end{matrix}\right] \right)\end{displaystyle}$, obtaining
\begin{equation}\label{EQN_LAMBDA} 
\bm{\Lambda}:=  \bold{QP}^{\intercal}  \bold{\Omega} \bold{Q} \bold{P}  =  \bold{P}^{\intercal} \bold{M} \bold{P}= \diag\left(\underbrace{\underbrace{\lambda_1, \lambda_2, \dots, \lambda_1, \lambda_2}_{n-1 {\ \rm times}}, \lambda_3, \lambda_4, \dots,\underbrace{\lambda_1, \lambda_2, \dots, \lambda_1, \lambda_2}_{n-1 {\ \rm times}}, \lambda_3, \lambda_4}_{n {\ \rm {times}}} \right), \end{equation}
where 
\begin{align}\label{EQN:GEN_EIG}
 \lambda_1 &= \frac{1}{\beta-\eta}, \qquad   \lambda_2= \frac{1}{\beta+\eta}, \qquad 
 \lambda_3 = {\frac {\alpha+\mu}{\alpha\,\gamma-\alpha\,\tau-{\delta}^{2}-2\,\nu\, \delta+\gamma\,\mu-\mu\,\tau-{\nu}^{2}}}, \\ \quad & \mbox{and} \quad  \lambda_4 = {\frac {\alpha-\mu}{\alpha\,\gamma+\alpha\,\tau-{\delta}^{2}+2\,\nu\,
\delta-\gamma\,\mu-\mu\,\tau-{\nu}^{2}}}.  \nonumber
\end{align}
\end{lemma}

We now begin substituting in the values of $\alpha$ through $\tau$ given in
(\ref{ALPHA_SO}-\ref{TAU_SO}) and (\ref{ALPHA_ISO}-\ref{TAU_ISO}) into the
results we have obtained for the $\bold{C}$ in order to
derive the results we need for $\bold{C}_{n,d}$ and $\bold{C}_n$.

\begin{lemma} \label{LEM:POSITIVE_SPECIAL_POINTS}
For all $t > 0$, $\bold{C}_n$ is positive definite.  For $d \geq 3$ and sufficiently small $t > 0$, $\bold{C}_{n,d}$ is positive definite.
\end{lemma}

\begin{proof}
It is a general fact from probability theory that the
covariance matrix of a random vector is positive semi-definite.
For $\bold{C}_n$, equation (\ref{DETC_GENERAL}) becomes
\begin{eqnarray*}
\det(\bold{C}_n)  = \left( {{\rm e}^{\frac{1}{2}\,{t}^{2}}}-{{\rm e}^{-\frac{1}{2}\,{t}^{2}}} \right) ^{n(n-1)} 
\left( {{\rm e}^{\frac{1}{2}\,{t}^{2}}}-{{\rm e}^{-\frac{1}{2}\,{t}^{2}}}+{t}^{2}
 \right)^n  \left( {{\rm e}^{\frac{1}{2}\,{t}^{2}}}-{{\rm e}^{-\frac{1}{2}\,{t}^{2}}}-{t }^{2} \right)^n,
\end{eqnarray*}
which is positive for all $t > 0$.

For $\bold{C}_{n,d}$ we have
\begin{eqnarray*}
\det(\bold{C}_{n,d}) = \frac{d^{2n^2+n}(d-1)^{n^2+n}(d-2)^n}{12^n} \, t^{2n^2+6n} + O(t^{2n^2+6n+1}),
\end{eqnarray*}
which is positive for $d \geq 3$ and $t > 0$ sufficiently small.
\end{proof}

\begin{lemma} $\bold{\Omega}_{n,d}$ and $\bold{\Omega}_n$ are orthogonally
diagonalized by $\bold{Q}\bold{P}$ where $\bold{Q}$ and $\bold{P}$ are the
$2n^2 \times 2n^2$ matrices described in Lemma \ref{DiagonalizingMatrix}.
The eigenvalues $\lambda_{d,1},\lambda_{d,2},\lambda_{d,3},\lambda_{d,3}$ of $\bold{\Omega}_{n,d}$ satisfy
\begin{align}\label{EQN:EIG_EXPANSIONS_FINITE_D}
\lambda_{1,d}^{-1/2} =  \sqrt{\frac{d(d-1)}{2}}t + O(t^3) \quad \lambda_{2,d}^{-1/2} =  \sqrt{2d} + O(t) \quad
\lambda_{3,d}^{-1/2} = \sqrt{d(d-1)} \, t  \quad \lambda_{4,d}^{-1/2} = \sqrt{{\frac {d({d}^{2}-3\,d+2)}{12 }}} \,  t^2 + O(t^3).
\end{align}
The eigenvalues $\lambda_1,\lambda_2,\lambda_3,\lambda_4$ of $\bold{\Omega}_n$ equal
\begin{align}\label{EQN:EIG_GAF}
\lambda_1 = \left( {\rm e}^{\frac{t^2}{4}}-{\rm e}^{-\frac{t^2}{4}} \right) ^{-1}, \quad
\lambda_2=  \left( {\rm e}^{\frac{t^2}{4}}+{\rm e}^{-\frac{t^2}{4}} \right) ^{-
1}
, \quad
\lambda_3 = \frac {{\rm e}^{\frac{t^2}{4}}+{\rm e}^{-\frac{t^2}{4}}}{{t}^{2}+
{\rm e}^{\frac{t^2}{2}}-{\rm e}^{-\frac{t^2}{2}}}
, \quad \mbox{and} \quad
\lambda_4 =  \frac {{\rm e}^{\frac{t^2}{4}}-{\rm e}^{-\frac{t^2}{4}}}{-{t}^{2}+
{\rm e}^{\frac{t^2}{2}}-{\rm e}^{-\frac{t^2}{2}}},
\end{align}
and satisfy
\begin{align}\label{EQN:EIG_EXPANSIONS}
\lambda_1^{-1/2} = \frac{t}{\sqrt{2}}+O \left( {t}^{3} \right) \quad \lambda_2^{-1/2} =  \sqrt {2}+O \left( {t}^{2} \right) \quad
\lambda_3^{-1/2} = t+O \left( {t}^{3} \right) \quad \lambda_4^{-1/2} = \frac{1}{\sqrt{12}}t^2 + O \left( {t}^{3} \right).
\end{align}
\end{lemma}

\begin{proof}
This follows from Lemma \ref{DiagonalizingMatrix} and by substituting the values of
$\alpha$ through $\tau$ from (\ref{ALPHA_SO}-\ref{TAU_SO}) and (\ref{ALPHA_ISO}-\ref{TAU_ISO}) into (\ref{EQN:GEN_EIG}).
(The asymptotics in (\ref{EQN:EIG_EXPANSIONS_FINITE_D}) determined using the Maple computer algebra system \cite{MAPLE}.  However, they
are simple enough that one can check them by hand.)
\end{proof}

We will need the following calculation in Section \ref{SEC:SHORT_RANGE}:
\begin{lemma} We have
\begin{align}  \label{EIG_AND_DET}
\frac{\left(\lambda_{1,d}\lambda_{2,d}\right)^{-\frac{1}{2}n\left(n-1\right)}
\left(\lambda_{3,d}\lambda_{4,d}\right)^{-\frac{1}{2} n}}{\sqrt{{\rm det} \ \bold{C}_{d,n}}} = d^{-\frac{n}{2}} \, t^{-n} + O(t^{-n+2}) \quad \mbox{and} \quad
\frac{\left(\lambda_1\lambda_2\right)^{-\frac{1}{2}n\left(n-1\right)} \left(\lambda_3\lambda_4\right)^{-\frac{1}{2} n}}{\sqrt{{\rm det} \ \bold{C}_n}} 
 = t^{-n} + O(t^{-n+2}).
\end{align}
\end{lemma}

\begin{proof}
Using (\ref{DETC_GENERAL}) and Lemma \ref{DiagonalizingMatrix} for the generic form of the covariance matrix $\bold{C}$ we have
\begin{align}
\frac{\left(\lambda_1\lambda_2\right)^{-\frac{1}{2}n\left(n-1\right)} \left(\lambda_3\lambda_4\right)^{-\frac{1}{2} n}}{\sqrt{{\rm det} \ \bold{C}_n}} =
 (\alpha^2 - \mu^2)^{-\frac{n}{2}}.
\end{align}
The result then follows by substituting (\ref{ALPHA_SO}-\ref{TAU_SO}) and (\ref{ALPHA_ISO}-\ref{TAU_ISO}) and doing an expansion.
\end{proof}

\section{Proof of Part (1) from Theorem \ref{THM:LOCAL}: short-range asymptotics} \label{SEC:SHORT_RANGE}
We apply the Kac-Rice formula to the covariance matrices $\bold{C}_{n,d}$ and
$\bold{C}_n$ and the submatrices $\bold{\Omega}_{n,d}$ and $\bold{\Omega}_n$ of
their inverses, as computed in Section \ref{SEC:COVARIANCE}.  It applies because, by Lemma \ref{LEM:POSITIVE_SPECIAL_POINTS},
$\bold{C}$ is positive definite for all $t > 0$ and $\bold{C}_{n,d}$ is positive definite for all $d \geq 3$ and sufficiently small $t > 0$.  The proof will be
the nearly same for each, so we will work with $\K_{n,d}(t)$ and then explain
what change needs to be made for $\K_n(t)$ at the very end of the section.

We apply the diagonalization of $\bold{\Lambda} \equiv \bold{\Lambda}_{n,d} =
(\bold{Q}\bold{P})^T \bold{\Omega}_{n,d} (\bold{Q} \bold{P})$ from
(\ref{Omega}) and (\ref{EQN_LAMBDA})  to the Kac-Rice formula (\ref{K1}) to
obtain \begin{equation} \label{KDiagonalized} \K_{n,d} \left(t\right)
=\frac{1}{\left(2\pi\right)^{n\left(n+1\right)}\rho_d(\bm{x})\rho_d(\bm{y})\sqrt{\det\
\bold{C}_{n,d}}}  \int \limits_{\mathbb{R}^{2n^2}} |\det \ {\bm{
\xi\left({\tau}\right)}}| |\det \ {\bm{ \eta\left({\tau}\right)}}| \
\mathrm{e}^{-\frac{1}{2} \left({\bm{\Lambda}}
{\bm{\tau}},{\bm{\tau}}\right)} d{\bm{\tau}}.  \end{equation} where
$\begin{displaystyle} \bm{\tau} := \left[ \begin {array}{ccccccc} {
\bm{\tau}_{1}}&{ \bm{\tau}_{2}}&\dots&{ \bm{\tau}_{n}}\end {array}
\right]^\intercal:=\bold{P}^{\intercal} \bold{Q}^{\intercal} {\bm{u}}
\end{displaystyle}$, where $\bm{\tau}_{i}=\left[ \begin {array}{ccccccc} {
{\tau}_{i,1}}&{ {\tau}_{i,2}}&\dots&{ {\tau}_{i,2n}}\end {array} \right]$ for each $1\leq i \leq n$.  In
these new variables, $\bm{\xi}$ and $\bm{\eta}$ become the new
matrices $\bm{  \xi\left({\tau}\right)}$ and $\bm{
\eta\left({\tau}\right)}$, whose entries are defined by \begin{equation}
\xi_{i,j}\left(\tau\right)= \frac{\sqrt{2}}{2}\left(-\tau_{i,2j-1}+
\tau_{i,2j}\right) \quad \mbox{and} \quad
\eta_{i,j}\left(\tau\right)=\frac{\sqrt{2}}{2}\left(\tau_{i,2j-1}+
\tau_{i,2j}\right) \quad \mbox{for} \quad i,j\leq n. \end{equation}

The reason for diagonalizing $\bm{\Omega}_{n,d}$ was to change the exponent into a form conducive to forming $n$ sets of $2n$-dimensional spherical coordinates so that $\left(\bm{\Lambda} \bm{\tau},\bm{\tau}\right)$ becomes $\sum\limits_{k=1}^n r_k^2$. \\ Let $\bm{w}:=[\begin{array} {cccccccc} r_1 & r_2 & \dots & r_n & \bm{\theta}_1 & \bm{\theta}_2 & \dots & \bm{\theta}_n \end{array}]$, where $\bm{\theta}_i=[\begin{array}{cccc} \theta_{i,1} & \theta_{i,2} & \dots & \theta_{i,2n-1}\end{array}]$. Let

\begin{equation} \label{SphericalChange}                                             
\tau_{i,j}=\begin{cases} \lambda_{1,d}^{-\frac{1}{2}}r_i \left( \prod\limits_{k=1}^{\frac{j-1}{2}} \sin \theta_{i,k} \right) \left(\cos \theta_{i,\frac{j+1}{2}}\right)  & \mbox{if } j \mbox{ is odd} \ \mbox{and } j \ne 2n-1,\   \\
 \lambda_{4,d}^{-\frac{1}{2}} r_i \left( \prod\limits_{k=1}^{n-1} \sin \theta_{i,k} \right)\left(\cos \theta_{i,n}  \right) & \mbox{if } j = 2n,\\
\lambda_{2,d}^{-\frac{1}{2}}r_i  \left( \prod\limits_{k=1}^{n-1+\frac{j}{2}} \sin \theta_{i,k} \right) \left(\cos \theta_{i,n+\frac{j}{2}} \right)  & \mbox{if } j \mbox{ is even} \ \mbox{and } j \ne 2n, \mbox{and}  \\
 \lambda_{3,d}^{-\frac{1}{2}}r_i\left( \prod\limits_{k=1}^{2n-1} \sin \theta_{i,k} \right) &\mbox{if } j=2n-1. \ 
 \end{cases}
 \end{equation}
 
Thus, $d\bm{\tau}$ becomes $\begin{displaystyle} \left(\lambda_{1,d}\lambda_{2,d}\right)^{-\frac{1}{2}n\left(n-1\right)} \left(\lambda_{3,d}\lambda_{4,d}\right)^{-\frac{1}{2} n}\prod \limits_{l=1}^n r_{l}^{2n-1} \  d\mu_{\left(\mathbb{S}^{2n-1}\right)^n} d \bm{r}\end{displaystyle}$ , where $\mathbb{S}^{2n-1}$ denotes the unit sphere in $\mathbb{R}^{2n}$ and $d\mu_{\left(\mathbb{S}^{2n-1}\right)^n}$ denotes the product
measure on $\left(\mathbb{S}^{2n-1}\right)^n$ obtained from the standard spherical measure $d\mu_{\mathbb{S}^{2n-1}}$ on $\mathbb{S}^{2n-1}$.  Let $\phi_{i,j}$ be the trigonometric product in $\tau_{i,j}$ so that $\begin{displaystyle} \tau_{i,j}=\lambda_{h\left(j\right),d}^{-\frac{1}{2}} r_i \phi_{i,j} \end{displaystyle}$. 
After this variable change, we see that  \begin{equation}\label{RSubs}   \xi_{i,j}\left(\bm{r},\bm{\theta}\right)=\frac{\sqrt{2}}{2} r_i \left(- \lambda_{m,d}^{-\frac{1}{2}}  \phi_{i,2j-1} +  \lambda_{m+1,d}^{-\frac{1}{2}}\phi_{i,2j} \right) {\ \ \rm and \ \ \ }  \eta_{i,j}\left(\bm{r},\bm{\theta}\right)= \frac{\sqrt{2}}{2} r_i \left(\lambda_{m,d}^{-\frac{1}{2}} \phi_{i,2j-1} +  \lambda_{m+1,d}^{-\frac{1}{2}}\phi_{n,2j} \right)\end{equation} where $m=1$ when $j \ne n$ and $m=3$ when $j=n$. \\

Thus, in the new spherical coordinates, we have: 
\begin{equation}  \label{Spherical} 
\K_{n,d}\left(t\right) =
\frac{\left(\lambda_{1,d}\lambda_{2,d}\right)^{-\frac{1}{2}n\left(n-1\right)}
\left(\lambda_{3,d}\lambda_{4,d}\right)^{-\frac{1}{2}
n}}{\left(2\pi\right)^{n\left(n+1\right)}
\rho_d(\bm{x})\rho_d(\bm{y})\sqrt{\det\  \bold{C}_{n,d}}}
\int\limits_{\mathbb{R}^{2n^2}} \mid {\rm det \ } \bm{
\xi}{\left(\bm{r},\bm{\theta}\right)} \mid \mid {\rm det \ } \bm{
\eta}{\left(\bm{r},\bm{\theta}\right)} \mid \mathrm{e}^{-\frac{1}{2}
\sum\limits_{k=1}^n r_k^2} \prod \limits_{l=1}^n r_{l}^{2n-1}
d\mu_{\left(\mathbb{S}^{2n-1}\right)^n} d\bm{r}. 
\end{equation}

Using the asymptotic behavior of $\lambda_{1,d}^{-1/2}, \lambda_{2,d}^{-1/2},
\lambda_{3,d}^{-1/2},$ and $\lambda_{4,d}^{-1/2}$ expressed in
(\ref{EQN:EIG_EXPANSIONS_FINITE_D}), we notice
that each the elements of the $n$th column of each determinant vanishes
linearly with $t$. Therefore, we factor out $t$ from each column to prevent
these columns from vanishing in the limit as $t$ goes to~$0$. Also note that
each element of row $i$ in both $\bm{\xi}$ and $\bm{\eta}$ are linear with
$r_i$.  Thus, we let $
\bm{\hat{\xi}}\left(\bm{\theta},t\right)$ and $
\bm{\hat{\eta}}\left(\bm{\theta},t\right)$ denote the resulting
matrices when $t$ is factored from the $n$th column and $r_i$ is factored from
each row of each matrix, $\bm{\xi}$ and $\bm{\eta}$, respectively.  Using
Fubini's Theorem, we can now split the integral from (\ref{Spherical}) into an
integral over the radii and an integral over the angles:

\begin{align} \label{SplitIntegrals}
\K_{n,d}\left(t\right) = \frac{\left(\lambda_{1,d}\lambda_{2,d}\right)^{-\frac{1}{2}n\left(n-1\right)}
\left(\lambda_{3,d}\lambda_{4,d}\right)^{-\frac{1}{2}
n}}{\left(2\pi\right)^{n\left(n+1\right)}
\rho_d(\bm{x})\rho_d(\bm{y})\sqrt{\det\  \bold{C}_{n,d}}}\left(\int\limits_{\mathbb{R}_{\geq 0}^n} \prod \limits_{l=1}^n r_{l}^{2n+1} \mathrm{e}^{-\frac{1}{2}\sum\limits_{k=1}^n \left(r_k^2\right)} d\bm{r} \right) \\ \cdot  \left( t^{2}\!\!\!\!\!\! \int\limits_{\left(\mathbb{S}^{2n-1}\right)^n}\!\!\!\!\!\! \mid \mbox{det } \hat{\bm{\xi}}\left(\bm{\theta},t\right)\mid \mid \mbox{det } \hat{\bm{\eta}}\left(\bm{\theta},t\right)\mid  \ d\mu_{\left(\mathbb{S}^{2n-1}\right)^n} \right).
\end{align}
Using the definition of the Gamma function, (\ref{SplitIntegrals}) simplifies to:
\begin{align}
\K_{n,d}\left(t\right)= \left(\frac{ 2^{n^2} \Gamma\left(n+1 \right)^n }{\left(2\pi\right)^{n\left(n+1\right)}
\rho_d(\bm{x})\rho_d(\bm{y})} \right) \left( \frac{\left(\lambda_{1,d}\lambda_{2,d}\right)^{-\frac{1}{2}n\left(n-1\right)}
\left(\lambda_{3,d}\lambda_{4,d}\right)^{-\frac{1}{2}
n}}{\sqrt{\det\  \bold{C}_{n,d}}} t^{2} \right) \\ \cdot  \left( \int\limits_{\left(\mathbb{S}^{2n-1}\right)^n} \mid \mbox{det } \hat{\bm{\xi}}\left(\bm{\theta},t\right)\mid \mid \mbox{det } \hat{\bm{\eta}}\left(\bm{\theta},t\right)\mid  \ d\mu_{\left(\mathbb{S}^{2n-1}\right)^n} \right). \nonumber
\end{align}

By (\ref{EIG_AND_DET}) we have that 
\begin{equation} \label{TaylorExpansionE}  \frac{\left(\lambda_{1,d}\lambda_{2,d}\right)^{-\frac{1}{2}n\left(n-1\right)}
\left(\lambda_{3,d}\lambda_{4,d}\right)^{-\frac{1}{2}
n}}{\sqrt{\det\  \bold{C}_{n,d}}} t^{2}  =  d^{-\frac{n}{2}} \, t^{2-n} + O\left(t^{4-n}\right).
\end{equation}
Meanwhile, by (\ref{EQN:EIG_EXPANSIONS_FINITE_D}), each entry of $\hat{\bm{\xi}}$ and $\hat{\bm{\eta}}$ is of the form: constant (potentially 0) plus $O(t)$. Therefore,
\begin{equation}
\int\limits_{\left(\mathbb{S}^{2n-1}\right)^n} \mid \mbox{det }\hat{\bm{\xi}}\left(\bm{\theta},t\right)\mid \mid \mbox{det }\hat{\bm{\eta}}\left(\bm{\theta},t\right)\mid  \ 
d\mu_{\left(\mathbb{S}^{2n-1}\right)^n} = D_n + O\left(t\right),
\end{equation}
where
\begin{equation} \label{LimitCalc} D_n = \lim\limits_{t\rightarrow 0} \int\limits_{\left(\mathbb{S}^{2n-1}\right)^n} \mid \mbox{det }\hat{\bm{\xi}}\left(\bm{\theta},t\right)\mid \mid \mbox{det }\hat{\bm{\eta}}\left(\bm{\theta},t\right)\mid  \ d\mu_{\left(\mathbb{S}^{2n-1}\right)^n}. \end{equation}
Finally,
we also have from (\ref{Density_FINITE_D}) that 
\begin{equation} \rho_d(\bm{x}) = \rho_d(\bm{y}) = \pi^{-\frac{n+1}{2}} \Gamma\left( \frac{n+1}{2} \right) d^{\frac{n}{2}} + O(t^2).\end{equation}

Therefore, 
\begin{align}
\K_{n,d}(t) = A_{n,d}\ t^{2-n} + O(t^{4-n}) \qquad \mbox{where} \qquad
A_{n,d} = \left(\frac{ 2^{n^2}  \pi^{n+1} \Gamma\left(n+1 \right)^n }{\left(2\pi\right)^{n\left(n+1\right)}
\Gamma\left( \frac{n+1}{2} \right)^2 d^{\frac{3}{2} n}}  \right) D_n.
\end{align}

We now compute the constant $D_n$.
From Equations (\ref{RSubs}), (\ref{EQN:EIG_EXPANSIONS_FINITE_D}) and (\ref{EQN:EIG_EXPANSIONS}):
   \begin{align}
\lim\limits_{t\rightarrow 0} \hat{\xi}_{i,j}\left(\bm{\theta},t\right) &= \lim\limits_{t\rightarrow 0} \hat{\eta}_{i,j}\left(\bm{\theta},t\right) = \sqrt{d} \,\, \phi_{i,2j}  \quad \quad \quad \ \ \ \text{for } j<n , \ \text{and} \\ \lim\limits_{t\rightarrow 0} -\hat{\xi}_{i,j}\left(\bm{\theta},t\right) &= \lim\limits_{t\rightarrow 0} \hat{\eta}_{i,j}\left(\bm{\theta},t\right) = \sqrt{\frac{d(d-1)}{2}} \,\, \phi_{i,2n-1} \quad \text{for } j=n. \label{MODIFIED_EQN}
\end{align}
Let $\begin{displaystyle}
\bm{\mu}\left(\bm{\theta}\right)\end{displaystyle}$ be the
resulting matrix when $\sqrt{d}$ is factored out of the first through $n-1$-st columns and  $\sqrt{\frac{d(d-1)}{2}}$ is factored out of the $n$th column
of $\lim\limits_{t\rightarrow 0}
\hat{\bm{\xi}}\left(\bm{\theta},t\right)$. If we do the same
process with $\lim\limits_{t\rightarrow 0}
\hat{\bm{\mu}}\left(\bm{\theta},t\right)$, we obtain the same
result with the sign changed in the $n$th column. Therefore,
\begin{equation}\label{LimitResults} D_n =  \frac{d^n (d-1)}{2} \,
\int\limits_{\left(\mathbb{S}^{2n-1}\right)^n} |\mbox{det }
\bm{\mu}\left(\bm{\theta}\right)|^2 \
d\mu_{\left(\mathbb{S}^{2n-1}\right)^n}. \end{equation}

From (\ref{SphericalChange}), we notice that each entry of row $i$ of $\bm{\mu}\left(\bm{\theta}\right)$ contains a factor of $\begin{displaystyle}\prod\limits_{j=1}^{n} \sin \theta_{i,j} \end{displaystyle}$. Thus, we can take this factor out of each row of the matrix, removing any dependence of the determinant on $\theta_{i,j}$ for $j\leq n$. Let $\bm{\nu}\left(\bm{\theta}\right)$ denote the matrix that remains after removing these factors.

We can then split the integral into two, an integral over $\mathbb{B}^n$, where $\mathbb{B}:= [0,\pi]$, corresponding to $\theta_{i,j}$ for $j\leq n$, and an integral over $\left(\mathbb{S}^{n-1}\right)^{n}$, corresponding to $\theta_{i,j}$ for $j> n$:
\begin{equation} \label{SplitIntegral2}
D_n =  \frac{d^n (d-1)}{2} \,  \int\limits_{\mathbb{B}^n} \prod\limits_{i,j\leq n} \left( \sin \theta_{i,j}\right)^{2n+1-j} \ d {\rm Leb}_{\mathbb{B}^n} \int\limits_{\left(\mathbb{S}^{n-1}\right)^n} {\mid \mbox{det } \bm{\nu}\left(\bm{\theta}\right) \mid}^2\ d\mu_{\left(\mathbb{S}^{n-1}\right)^n}.
\end{equation}
Here we have used that $\begin{displaystyle} d\mu_{\mathbb{S}^{2n-1}} = \prod\limits_{j=1}^{n} \sin \theta_j^{2n-1-j}  d{\rm Leb}_{\mathbb{B}} \ d\mu_{\mathbb{S}^{n-1}}. \end{displaystyle}$
The former integral in this product can be calculated recursively with integration by parts to be:
\begin{equation}\label{FormerIntegral}
\int\limits_{\mathbb{B}^n} \prod\limits_{i,j\leq n} \left( \sin \theta_{i,j}\right)^{2n+1-j} \ 
d{\rm Leb}_{\mathbb{B}^n} = \left(\prod\limits_{j=1}^n \left(\int\limits_0^{\pi} \left( \sin \theta \right)^{2n+1-j} d\theta \right)\right)^n = \left(\pi ^{\lceil{\frac{n}{2}  \rceil}} 2^{\lfloor{\frac{n}{2}} \rfloor} \frac{n!!}{2n!!}\right)^n.
\end{equation}
The calculation of the latter integral follows from Proposition \ref{kthMoment} of Appendix \ref{kthMomentProof}:
\begin{equation} \label{LatterIntegral}
\int\limits_{\left(\mathbb{S}^{n-1}\right)^n} {\mid \mbox{det } \bm{\nu} \mid}^2\ 
d\mu_{\left(\mathbb{S}^{n-1}\right)^n} = \left(\frac{\Gamma\left(\frac{n}{2}\right)}{\Gamma\left(\frac{n+2}{2}\right)}\right)^{n-1} \frac{\Gamma\left(\frac{n+1}{2}\right)\Gamma\left(\frac{n}{2}\right)}{\Gamma\left(\frac{1}{2}\right)} \left( \frac{n \pi^{\frac{n}{2}}}{ \Gamma\left(\frac{n+2}{2}\right)} \right)^n  = \frac{\Gamma\left(n+1\right)\pi^{\frac{n^2}{2}}}{\Gamma\left(\frac{n+2}{2}\right)^n}.
\end{equation}

Therefore,
\begin{align}
A_{n,d} =   \left(\frac{ 2^{n^2}  \pi^{n+1} \Gamma\left(n+1 \right)^n }{\left(2\pi\right)^{n\left(n+1\right)}
\Gamma\left( \frac{n+1}{2} \right)^2} \right) D_n
&=   \left(\frac{ 2^{n^2}  \pi^{n+1} \Gamma\left(n+1 \right)^n }{\left(2\pi\right)^{n\left(n+1\right)}
\Gamma\left( \frac{n+1}{2} \right)^2 d^{\frac{3}{2} n}} \right) \left( \frac{d^n (d-1)}{2} \right)
\left(\pi^{ \lceil{\frac{n}{2}  \rceil}} 2^{ \lfloor{\frac{n}{2} \rfloor}}
\frac{n!!}{2n!!}\right)^n \frac{\Gamma\left(n+1\right)\pi^{\frac{n^2}{2}}}{\Gamma\left(\frac{n+2}{2}\right)^n} \\
&= \left(\frac{d-1}{d^{\frac{n}{2}}}\right) \frac{\sqrt{\pi}\ \Gamma\left(\frac{n+2}{2}\right)}{2\ \Gamma\left( \frac{n+1}{2} \right)} = C_{n,d}.
\end{align}
Thus, 
\begin{eqnarray*}
\mathcal{K}_{n,d}\left(t\right) = \left(\frac{d-1}{d^{\frac{n}{2}}}\right) \frac{\sqrt{\pi}\ \Gamma\left(\frac{n+2}{2}\right)}{2\ \Gamma\left( \frac{n+1}{2} \right)} t^{2-n} + O\left(t^{3-n}\right),
\end{eqnarray*}
 as stated in Part (1) from Theorem \ref{THM:LOCAL}.

The only differences when computing $\K_n(t)$ instead of $\K_{n,d}(t)$  are:
\begin{enumerate}
\item $\lambda_1^{-1/2} = \sqrt{2} + O(t^2)$ instead of $\sqrt{2d} + O(t^2)$ and $\lambda_3^{-1/2} = t + O(t^2)$ rather than  $\lambda_3^{-1/2} = \sqrt{\frac{d-1}{d}}t + O(t^2)$,
\item the factor of $d^{-\frac{1}{2} n}$ in (\ref{EIG_AND_DET}) is missing, and 
\item the factors of $d^{\frac{1}{2} n}$  are missing from the expression for the density of the zeros.
\end{enumerate}
One can readily check that 
this results
in the factor of $\frac{d-1}{d^{\frac{n}{2}}}$ being removed from the constant:
\begin{eqnarray*}
\mathcal{K}_{n}\left(t\right) =  \frac{\sqrt{\pi}\ \Gamma\left(\frac{n+2}{2}\right)}{2\ \Gamma\left( \frac{n+1}{2} \right)} t^{2-n} + O\left(t^{3-n}\right)
\end{eqnarray*}

\qed \, (Part (1) of Theorem \ref{THM:LOCAL}).

\section{Proof of Part (2) from Theorem \ref{THM:LOCAL}: long-range asymptotics} \label{SEC:LONG_RANGE}

It will be convenient to apply the Kac-Rice formulae to the ensemble $\bm{g}$
given in (\ref{G}), which has the same zeros as the $\isoR$-invariant ensemble
$\bm{f}$.  Let $\bold{C}_{n,\bm{g}}$ denote the covariance matrix applied to
random vector (\ref{Vector}) for this ensemble.  Recall that $\bm{x}$ and
$\bm{y}$ are given by (\ref{EQN:DEF_XY}).  The following covariances can be
computed from those in (\ref{ALPHA_ISO}-\ref{TAU_ISO}) and the product rule:
\begin{eqnarray}\label{GAF_COV_FIRST}
E\left(g(\bm{x}) g(\bm{x})\right) &=& E\left(g(\bm{y}) g(\bm{y})\right)= 1,\label{GCOV_FIRST} \\
E\left(g(\bm{x}) \frac{\partial g(\bm{x})}{\partial x_j }\right) &=& E\left(g(\bm{y}) \frac{\partial g(\bm{y})}{\partial y_j }\right) = 0, \\
E\left(\frac{\partial g(\bm{x})}{\partial x_j }\frac{\partial g(\bm{x})}{\partial x_j }\right) &=& E\left(\frac{\partial g(\bm{y})}{\partial y_j }\frac{\partial g(\bm{y})}{\partial y_j }\right) = 1, \qquad  \\
E\left(\frac{\partial g(\bm{x})}{\partial x_i }\frac{\partial g(\bm{x})}{\partial x_j }\right) &=& E\left(\frac{\partial g(\bm{y})}{\partial y_i }\frac{\partial g(\bm{y})}{\partial y_j }\right) =  0 \qquad \mbox{if} \qquad i \neq j,  \label{GCOV_MIDDLE} \\
E\left(g(\bm{x})g(\bm{y})\right) &=& {\rm e}^{-\frac{1}{2}  ||\bm{x}-\bm{y}||^2} = {\rm e}^{-\frac{t^2}{2}}, \\
 E\left(g(\bm{x}) \frac{\partial g(\bm{y})}{\partial y_j }\right) &=& - E\left(g(\bm{y}) \frac{\partial g(\bm{x})}{\partial x_j }\right) = {\rm e}^{-\frac{1}{2}  ||\bm{x}-\bm{y}||^2} \left(x_j-y_j\right)
= \begin{cases} 0 & \mbox{if } j \neq n \\ - t e^{-\frac{t^2}{2}}  & \mbox{if } j = n.\end{cases}, \\
E\left(\frac{\partial g(\bm{x})}{\partial x_i }\frac{\partial g(\bm{y})}{\partial y_i }\right) &=& {\rm e}^{-\frac{1}{2}  ||\bm{x}-\bm{y}||^2}\left( 1-\left
(x_i-y_i\right)^2 \right)
= \begin{cases} {\rm e}^{-\frac{t^2}{2}}   \quad  \mbox{if } \quad i \neq n \\
 {\rm e}^{-\frac{t^2}{2}}(1-t^2)     \quad \mbox{if } \quad i = n
\end{cases}, \qquad \mbox{and} \\
E\left(\frac{\partial g(\bm{x})}{\partial x_j }\frac{\partial g(\bm{y})}{\partial y_k }\right) &=&  - {\rm e}^{-\frac{1}{2}  ||\bm{x}-\bm{y}||^2} \left(x_j-y_j\right)\left(x_k-y_k\right) =  0.  \label{GAF_COV_LAST}
\end{eqnarray}
Remark that $\bold{C}_{n,\bm{g}}$ has the structure asserted in (\ref{COV_MATRIX_CALC}) and that
\begin{eqnarray*}
\det(\bold{C}_{n,\bm{g}}) =  
\left( 1-{{\rm e}^{-{t}^{2}}} \right) ^{n(n-1)}
 \left(1+ {{\rm e}^{-\frac{1}{2}\,{t}^{2}}}{t}^{2}-{
{\rm e}^{-{t}^{2}}} \right)^n  \left( 1-{{\rm e}^{-\frac{1}{2}\,{t}^{2}}}{t}^{
2}-{{\rm e}^{-{t}^{2}}} \right)^n >  0
\end{eqnarray*}
for all $t > 0$, so that $\bold{C}_{n,\bm{g}}$ is positive definite.

The proof will rely upon two facts:
\begin{align}
\left(\det \bold{C}_{n,\bm{g}}\right)^{-\frac{1}{2}}  = 1+O(t^4 \mathrm{e}^{-t^2}) \quad \mbox{and} \label{FACT1} \\
||\bold{I} - \bm{\Omega}_{n,\bm{g}}||_{\infty} = O(t^2 \mathrm{e}^{-t^2}), \label{FACT2}
\end{align}
where $|| \ ||_{\infty}$ denotes the maximum entry of the matrix.
The former can be obtained from expression (\ref{DETC_GENERAL}). The latter 
follows from the calculations above
and Lemma
\ref{OmegaLem} expressing $\bold{\Omega}_{n,\bm{g}}$ in terms of the entries of
$\bold{C}_{n,\bm{g}}$.

The covariance matrix for random vector (\ref{Vector_density}) is
the identity, by (\ref{GCOV_FIRST}-\ref{GCOV_MIDDLE}) above. Thus,
the Kac-Rice Formula for the density of zeroes of $\bm{g}(\bm{x})$ gives
\begin{align} 1 &= \frac{ \rho(\bm{x})\rho(\bm{y})}{ \rho(\bm{x})\rho(\bm{y})}=\frac{1}{\rho(\bm{x})\rho(\bm{y})} \left(  \frac{1}{(2\pi)^{\frac{n(n+1)}{2}}} \int\limits_{\mathbb{R}^{n^2}} |\det \bm{\xi} | \mathrm{e}^{-\frac{1}{2}(\bm{\xi},\bm{\xi})} d\bm{\xi}\right) \left( \frac{1}{(2\pi)^{\frac{n(n+1)}{2}}} \int\limits_{\mathbb{R}^{n^2}} |\det \bm{\eta} | \mathrm{e}^{-\frac{1}{2}(\bm{\eta},\bm{\eta})} d\bm{\eta}\right) \nonumber \\ &=\label{IIntegral} \frac{1}{(2\pi)^{n(n+1)} \rho(\bm{x})\rho(\bm{y})} \int\limits_{\mathbb{R}^{2n^2}} | \det \bm{\xi}| |\det \bm{\eta}| \mathrm{e}^{-\frac{1}{2} \left(\bm{u},\bm{u}\right)} d\bm{u}. \end{align}

Since $\bold{C}_{n,\bm{g}}$ is positive definite,
the Kac-Rice formula for two-point correlations (\ref{K1}) and Equation (\ref{IIntegral}) give
\begin{align} |\mathcal{K}_n(t)-1| &= \left|\left(\frac{1}{\left(2\pi\right)^{n\left(n+1\right)} \rho(\bm{x})\rho(\bm{y})\sqrt{\det \bold{C}_{n,\bm{g}}}} \int\limits_{\mathbb{R}^{2n^2}} | \det \bm{\xi}| |\det \bm{\eta}| \mathrm{e}^{-\frac{1}{2} \left(\bm{\Omega}_{n,\bm{g}}\bm{u},\bm{u}\right)} d\bm{u}\right) - 1\right| \nonumber \\ 
&= \frac{1}{(2\pi)^{n(n+1)} \rho(\bm{x})\rho(\bm{y})\sqrt{\det \bold{C}_{n,\bm{g}}}} \\ &\cdot  \left| \, \int\limits_{\mathbb{R}^{2n^2}} | \det \bm{\xi}| |\det \bm{\eta}| \mathrm{e}^{-\frac{1}{2} \left(\bm{\Omega_{n,\bm{g}}}\bm{u},\bm{u}\right)} d\bm{u}  - \sqrt{\det \bold{C}_{n,\bm{g}}} \int\limits_{\mathbb{R}^{2n^2}} | \det \bm{\xi}| |\det \bm{\eta}| \mathrm{e}^{-\frac{1}{2} \left(\bm{u},\bm{u}\right)} d\bm{u} \right| \nonumber \\
&\leq \frac{\sqrt{\det \bold{C}_{n,\bm{g}}}}{(2\pi)^{n(n+1)} \rho(\bm{x})\rho(\bm{y})} \left| \,  \int\limits_{\mathbb{R}^{2n^2}}| \det \bm{\xi}| |\det \bm{\eta}| \mathrm{e}^{-\frac{1}{2} \left(\bm{\Omega_{n,\bm{g}}}\bm{u},\bm{u}\right)} d\bm{u} - \int\limits_{\mathbb{R}^{2n^2}} | \det \bm{\xi}| |\det \bm{\eta}| \mathrm{e}^{-\frac{1}{2} \left(\bm{u},\bm{u}\right)} d\bm{u} \right| \label{PART1} \\
&+ \frac{\left|1 - (\det \bold{C}_{n,\bm{g}})^{-\frac{1}{2}} \right|}{(2\pi)^{n(n+1)} \rho(\bm{x})\rho(\bm{y})}  \int\limits_{\mathbb{R}^{2n^2}} | \det \bm{\xi}| |\det \bm{\eta}| \mathrm{e}^{-\frac{1}{2} \left(\bm{u},\bm{u}\right)} d\bm{u}. \label{PART2}
\end{align}
Equation (\ref{PART2}) is $O(t^4 \mathrm{e}^{-t^2})$ by (\ref{FACT1}).

From Lemma \ref{COCV3} Part 2, using $\bold{A} = \bold{I}$ and $\bold{B} = \bm{\Omega_{n,\bm{g}}}$, we have
\begin{align}
\left| \int\limits_{\mathbb{R}^{2n^2}} | \det \bm{\xi}| |\det \bm{\eta}| \mathrm{e}^{-\frac{1}{2} \left(\bm{\Omega}_{n,\bm{g}}\bm{u},\bm{u}\right)}  -  \int\limits_{\mathbb{R}^{2n^2}} | \det \bm{\xi}| |\det \bm{\eta}| \mathrm{e}^{-\frac{1}{2} \left(\bold{I}\bm{u},\bm{u}\right)} d\bm{u} \right| 
= O\left(||\bold{I}-\bm{\Omega}_{n,\bm{g}}||_{\infty}^{1/2}\right) = O\left(t \mathrm{e}^{-\frac{t^2}{2}}\right).
\end{align}

Therefore, $|\mathcal{K}(t) - 1| = O\left(t \mathrm{e}^{-\frac{t^2}{2}}\right)$, so we obtain the desired result.

\qed \, (Part (2) of Theorem \ref{THM:LOCAL}).


\section{Proof of Theorem \ref{THM:UNIVERSALITY}: Universality} \label{SEC:UNIVERSALITY}

This section is devoted to a proof of  Theorem
\ref{THM:UNIVERSALITY}.   It will be divided into three parts: 1)
Reduction to a local version, 2) Statement and proofs of two lemmas, and 3)
Proof of the local version.

\subsection{Reduction to a local version of Theorem \ref{THM:UNIVERSALITY}}\label{SUBSEC_REDUCTION}

Let the homogeneous coordinates on $\mathbb{RP}^k$ be denoted
$[Z_1,\ldots,Z_{k+1}]$.  After applying a suitable isometry from $SO(k+1)$ we
can assume that $p = [0:\cdots:0:1]$, allowing us to work in the affine (local)
coordinates
\begin{align}\label{LOCAL_COORDS_RPK}
z_1 = \frac{Z_1}{Z_{k+1}}, \quad \ldots, \quad z_k = \frac{Z_k}{Z_{k+1}}
\end{align}
having $[0:\cdots:0:1] \in \mathbb{RP}^k$ as their origin.

In these local coordinates, the tangent space $T_p(M)$ becomes an
$n$-dimensional linear subspace of $\mathbb{R}^k$ containing the two points
$\bm{x} \neq \bm{y}$.  We can now rotate about $\bm{0}$ in these local
coordinates by an element of $SO(k)$  \big(corresponding to an element of
$SO(k+1)$ that fixes $[0:\cdots:0:1]$\big) allowing us to assume that
\begin{enumerate}
\item $T_p M = {\rm span}({\bm e}_1,\ldots,{\bm e}_n)$, where ${\bm e}_1,\ldots,{\bm e}_k$ are the standard basis vectors on $\mathbb{R}^k$, and 
\item $\bm{x} = (0,\ldots,0,s,t)$ and $\bm{y} = (0,\ldots,0,u)$ in the local coordinates $(z_1,\ldots,z_n)$ on $T_p M \equiv \mathbb{R}^n$,
\end{enumerate}
where $(s,t) \neq (0,u)$.
Since the ensemble (\ref{System}) on $\mathbb{RP}^k$ is invariant under elements of
$SO(k+1)$ and since we have rotated the submanifold $M$ and the points $\bm{x}$
and $\bm{y}$ under the same composition of elements of $SO(k+1)$, the
correlation function remains the same.

By our choice 1, above, $M$ is locally expressed as a graph of a $C^2$ function $\bm{\psi}: \mathbb{R}^n \rightarrow \mathbb{R}^{k-n}$
that satisfies
\begin{eqnarray}\label{PSI_VANISHING_TO_ORDER2}
\bm{\psi}(\bm{0}) = \bm{0} \qquad \mbox{and} \qquad D\bm{\psi}(\bm{0}) = \bm{0}.
\end{eqnarray}
The orthogonal projection  $\proj_p: T_p M \equiv \mathbb{R}^n  \rightarrow M$
is given by $\proj_p({\bm w}) = \left(({\bm w}),{\bm \psi}({\bm w})\right)$ for any
${\bm w} \in \mathbb{R}^n$.

In affine coordinates (\ref{LOCAL_COORDS_RPK}), the $SO(k+1)$-invariant polynomials are
\begin{equation} \label{AffineEnsemble2}
f_d(\bm{z}) = \sum\limits_{|\bm{\alpha}| \leq d} \sqrt{\binom{d}{\bm{\alpha}}} a_{\bm{\alpha}} \bm{z}^{\bm{\alpha}} \quad \mbox{where} \quad \binom{d}{\bm{\alpha}} = \frac{d!}{(d-|\bm{\alpha}|)!\prod\limits_{i=1}^{n} \alpha_{i}!}.
\end{equation}
and the $a_{\bm{\alpha}}$ are iid on the standard normal distribution $\mathcal{N}\left(0, 1\right)$.

The correlation between zeros
\begin{eqnarray*}
K_{n,d,M}\left(\proj_p\left(\frac{\bm x}{\sqrt{d}}\right),\proj_p \left(\frac{\bm y}{\sqrt{d}}
\right)\right)
\end{eqnarray*}
is the same as the correlation between zeros for the pull-back of this ensemble to the tangent space $T_pM \equiv \mathbb{R}^n \subset \mathbb{R}^k$ under $\proj_p$, which is given
by systems of $n$ functions chosen iid of the form
\begin{eqnarray}\label{TANGENT_DIST}
h_{d,\bm{\psi}}\left(\frac{\bm{x}}{\sqrt{d}}\right) := f_d\left(\frac{\bm{x}}{\sqrt{d}},\bm{\psi}\left(\frac{\bm{x}}{\sqrt{d}}\right)\right).
\end{eqnarray}
This follows because one need not use round balls in the definition (\ref{DefK}) of the correlation function--any
sequence of neighborhoods that is sufficiently nice for computing a Radon-Nikodym derivative suffices (see Remark \ref{RMK_NBHDS}).  If one uses round balls
$N_\delta\left(\proj_p\left(\frac{\bm x}{\sqrt{d}}\right)\right)$ and $N_\delta\left(\proj_p\left(\frac{\bm y}{\sqrt{d}}\right)\right)$ in the definition of $K_{n,d,M}\left(\proj_p\left(\frac{\bm x}{\sqrt{d}}\right),\proj_p \left(\frac{\bm y}{\sqrt{d}}\right)\right)$, then their preimages under the $C^2$ mapping $\proj_p$
will be suitable neighborhoods for defining the correlation function for the pull-back (\ref{TANGENT_DIST}).

Thus, we have reduced the statement of Theorem \ref{THM:UNIVERSALITY} to:
\begin{thm}[\bf Local version of Theorem \ref{THM:UNIVERSALITY}]\label{LOCAL_UNIVERSALITY}
Let  $K_{n,d,\bm{\psi}} \equiv K_{d,\bm{\psi}}$ denote the correlation function for systems of $n$ functions chosen iid of the form
(\ref{TANGENT_DIST}) and let $K_n$ denote the correlation function for the $\isoR$-invariant system~(\ref{GAF}).  

For any $s,t,u \in \mathbb{R}$, if  $\bm{x} = (0,\ldots,0,s,t) \in \mathbb{R}^n$, $\bm{y} = (0,\ldots,0,u) \in \mathbb{R}^n$, and $\bm{x} \neq \bm{y}$, then 
\begin{eqnarray*}
K_{n,d,\bm{\psi}}\left(\left(\frac{\bm{x}}{\sqrt{d}}\right),\left(\frac{\bm y}{\sqrt{d}} \right)\right) =  K_n\big(\bm{x}, \bm{y} \big) + O\left(\frac{1}{\sqrt{d}}\right),
\end{eqnarray*}
with the constant implicit in the $O$ notation depending uniformly on compact
subsets of $\mathbb{R}^2 \times \mathbb{R} \setminus
\{(s,t) = (0,u)\}$.
\end{thm}

We will need the following more detailed notation in the next two subsections.
Let $\bm{\psi}(\bm{x}) = (\psi_1(\bm{x}),\ldots,
\psi_{k-n}(\bm{x}))$.  If we write $\bm{\alpha} = (\bm{\beta},\bm{\gamma})$ with
$\bm{\beta} \in \mathbb{Z}_+^n$ and $\bm{\gamma} \in \mathbb{Z}_+^{k-n}$, then (\ref{TANGENT_DIST}) becomes
\begin{eqnarray} 
h_{d,\bm{\psi}} \left(\frac{\bm{x}}{\sqrt{d}}\right) &=& \sum\limits_{|\bm{\beta}| + |\bm{\gamma}| \leq d}
b_{(\bm{\beta},\bm{\gamma})}
\left(\frac{\bm{x}}{\sqrt{d}}\right)^{\bm{\beta}}
\left(\bm{\psi}\left(\frac{\bm{x}}{\sqrt{d}}\right)\right)^{\bm{\gamma}} \\
\quad \mbox{where} \quad b_{(\bm{\beta},\bm{\gamma})} &:=& \sqrt{\binom{d}{\bm{\beta} \,\, \bm{\gamma}}} a_{(\bm{\beta},\bm{\gamma})} \quad \mbox{and} \quad \binom{d}{\bm{\beta} \,\, \bm{\gamma}} = \frac{d!}{(d-|\bm{\beta}|-|\bm{\gamma}|)!\prod\limits_{i=1}^{n} \beta_{i}! \prod\limits_{i=1}^{k-n} \gamma_{i}!}. \nonumber
\end{eqnarray}
As before, the coefficients $a_{(\bm{\beta},\bm{\gamma})}$ are iid on the standard normal distribution $\mathcal{N}\left(0, 1\right)$.

\subsection{Two lemmas}

\begin{lemma}\label{GAF_POS_DEF}
For any $\bm{x} = (0,0,\ldots,s,t)$ and $\bm{y} = (0,0,\ldots,0,u)$ in $\mathbb{R}^n$ with $\bm{x} \neq \bm{y}$ we have:
\begin{enumerate}

\item 
The covariance matrix $\bold{C}$ corresponding to random vector
(\ref{Vector_density}) from the Kac-Rice formula for density (\ref{K_density}) applied to the $\isoR$-invariant ensemble~(\ref{GAF}) (at $\bm x$ or $\bm y$) is
positive definite.  The submatrix $\bold{\Omega}$ of $\bold{C}^{-1}$
defined in (\ref{K_density}) is also positive definite.

\item 
The covariance matrix $\bold{C}$ corresponding to random vector
(\ref{Vector}) from the Two Point Kac-Rice formula (\ref{K1}) applied to the $\isoR$-invariant ensemble~(\ref{GAF}) is
positive definite.  The submatrix $\bold{\Omega}$ of $\bold{C}^{-1}$
defined in (\ref{K1}) is also positive definite.
\end{enumerate}
\end{lemma}

\begin{proof}
We give the proof of Part 2 leaving the necessary modifications for Part 1 to the reader.

It is a general fact from probability theory that the
covariance matrix of a random vector is positive semi-definite.  Thus, it will
be sufficient for us to check that $\det(\bold{C}) > 0$.
We substitute 
$\bm{x} = (0,0,\ldots,s,t)$ and $\bm{y} = (0,0,\ldots,0,u)$ into the covariances computed in Equations 
(\ref{ISO_FIRST}-\ref{ISO_LAST}) 
from the proof of Lemma \ref{COV_MATRIX_CALC}, obtaining
that
$\bold{C}$ is of the form $\diag_n(\tilde{\bold{C}})$, where
 $\begin{displaystyle}\bold{\tilde{C}}= \left[ \begin{matrix}
\bold{A} & \bold{B}^{\intercal} \\
\bold{B}  & \bold{D} \\ \end{matrix} \right]
\end{displaystyle}$  and $\bold{A}, \bold{B}$, and $\bold{D}$ are the following $(n+1) \times (n+1)$ matrices:
 \begin{align}
\bold{A} =  e^{s^2+t^2} \left[ \begin{matrix}
1 &  0  & \dots & s & t \\
0  &  1 & \ddots &  & 0 \\
\vdots & \ddots & \ddots & \ddots & \vdots \\
s &   & \ddots & 1+s^2 & st \\
t & 0 & \dots & st & 1+t^2 \end{matrix}\right], \qquad
 \bold{B} =  e^{tu} \left[ \begin{matrix}
1 &  0  & \dots & 0 & u \\
0  &  1 & \ddots &  & 0 \\
\vdots & \ddots & \ddots & \ddots & \vdots \\
s &   & \ddots & 1 & su \\
t & 0 & \dots & 0 & 1+tu \end{matrix}\right], \quad \mbox{and}
\end{align}
\begin{align}
\bold{D} = e^{u^2} \left[ \begin{matrix}
1 &  0  & \dots & 0 & u \\
0  &  1 & \ddots &  & 0 \\
\vdots & \ddots & \ddots & \ddots & \vdots \\
0 &   & \ddots & 1 & 0 \\
u & 0 & \dots & 0 & 1+u^2 \end{matrix}\right].
\end{align}
After applying a suitable permutation to the rows and columns, $\bold{\tilde{C}}$
becomes a block-diagonal matrix with $n-2$ copies of
\begin{eqnarray*}
\left[\begin{matrix} e^{s^2+t^2}  & e^{tu} \\ e^{tu} & e^{u^2}\end{matrix}\right]
\end{eqnarray*}
 and one copy of
\begin{eqnarray}\label{WEIRD_DET}
\left[ \begin {array}{cccccc} {{\rm e}^{{s}^{2}+{t}^{2}}}&s{{\rm e}^{
{s}^{2}+{t}^{2}}}&t{{\rm e}^{{s}^{2}+{t}^{2}}}&{{\rm e}^{tu}}&s{
{\rm e}^{tu}}&t{{\rm e}^{tu}}\\ \noalign{\medskip}s{{\rm e}^{{s}^{2}+{
t}^{2}}}& \left(1+ {s}^{2} \right) {{\rm e}^{{s}^{2}+{t}^{2}}}&st{
{\rm e}^{{s}^{2}+{t}^{2}}}&0&{{\rm e}^{tu}}&0\\ \noalign{\medskip}t{
{\rm e}^{{s}^{2}+{t}^{2}}}&st{{\rm e}^{{s}^{2}+{t}^{2}}}& \left(1+ {t}^{
2} \right) {{\rm e}^{{s}^{2}+{t}^{2}}}&u{{\rm e}^{tu}}&su{{\rm e}^{t
u}}& \left(1+ tu \right) {{\rm e}^{tu}}\\ \noalign{\medskip}{{\rm e}^{
tu}}&0&u{{\rm e}^{tu}}&{{\rm e}^{{u}^{2}}}&0&u{{\rm e}^{{u}^{2}}}
\\ \noalign{\medskip}s{{\rm e}^{tu}}&{{\rm e}^{tu}}&su{{\rm e}^{tu}}&0
&{{\rm e}^{{u}^{2}}}&0\\ \noalign{\medskip}t{{\rm e}^{tu}}&0& \left(1+ t
u \right) {{\rm e}^{tu}}&u{{\rm e}^{{u}^{2}}}&0& \left(1+ {u}^{2}
 \right) {{\rm e}^{{u}^{2}}}\end {array} \right] 
\end{eqnarray}
The former has determinant $e^{2tu} (e^{s^2+(t-u)^2}-1)$, which is positive since $r = s^2+(t-u)^2 > 0$,
by our hypothesis that $\bm x \neq \bm y$.  
The latter has determinant equal to 
\begin{eqnarray*}
{{\rm e}^{6\,tu}} \left(  \left( {{\rm e}^{{s}^{2}+ \left( t-u
 \right) ^{2}}}-{{\rm e}^{2\,{s}^{2}+2\, \left( t-u \right) ^{2}}}
 \right)  \left(  \left( {s}^{2}+ \left( t-u \right) ^{2} \right) ^{2}
+3 \right) +{{\rm e}^{3\,{s}^{2}+3\, \left( t-u \right) ^{2}}}-1
 \right). 
\end{eqnarray*}
Without the exponential prefactor, this equals
$\left( {{\rm e}^{r}}-{{\rm e}^{2\,r}} \right)  \left( {r}^{2}+3
 \right) +{{\rm e}^{3\,r}}-1$,
which one can also check is positive for all $r > 0$.

Since $\bold{C}$ is positive definite, so is $\bold{C}^{-1}$.  After applying a suitable permutation to the rows and columns of $\bold{C}^{-1}$
the $n^2 \times n^2$ principal minor is 
$\bold{\Omega}$, which is therefore positive definite.
\end{proof}

\begin{lemma}\label{LEM_CONV_COV_MATRICES} 
For any $\bm{x} = (0,0,\ldots,s,t)$ and $\bm{y} = (0,0,\ldots,0,u)$ in $\mathbb{R}^n$ with $\bm{x} \neq \bm{y}$ we have:
\begin{enumerate}
\item 
Let $\bold{C}_{d,\bm{\psi}}$ and $\bold{C}$ be the covariance matrix for vector
(\ref{Vector_density}) applied to the systems $\bm{h}_{d,\bm{\psi}}$ given in  (\ref{TANGENT_DIST}) and $\bm{f}$ given in (\ref{GAF}), respectively, (at either $\bm x$ or $\bm y$) and let $\bold{\Omega}_{d,\bm{\psi}}$ and $\bold{\Omega}$
denote the submatrices of $\bold{C}_{d,\bm{\psi}}^{-1}$ and $\bold{C}^{-1}$  defined in the Kac-Rice formula for density (\ref{K_density}).
Then, 
\begin{align*}
\bold{C}_{d,\bm{\psi}} = \bold{C} + O\left(\frac{1}{d}\right) \qquad \mbox{and} \qquad \bold{\Omega}_{{d},\bm{\psi}} = \bold{\Omega} + O\left(\frac{1}{d}\right) ,
\end{align*} where the constants implicit in the notation depends uniformly on compact
subsets $\mathbb{R}^2$ (if we are working at $\bm x$) or $\mathbb{R}$ (if we are working at $\bm y$).
\item
Let $\bold{C}_{d,\bm{\psi}}$ and $\bold{C}$ be the covariance matrix for vector
(\ref{Vector}) applied to the systems $\bm{h}_{d,\bm{\psi}}$ given in  (\ref{TANGENT_DIST}) and $\bm{f}$ given in (\ref{GAF}), respectively, and let $\bold{\Omega}_{d,\bm{\psi}}$ and $\bold{\Omega}$
denote the submatrices of $\bold{C}_{d,\bm{\psi}}^{-1}$ and $\bold{C}^{-1}$  defined in the Kac-Rice formula for density (\ref{K1}).
Then,
\begin{align*}
\bold{C}_{d,\bm{\psi}} = \bold{C} + O\left(\frac{1}{d}\right) \qquad \mbox{and} \qquad \bold{\Omega}_{{d},\bm{\psi}} = \bold{\Omega} + O\left(\frac{1}{d}\right) 
\end{align*} 
with the constant implicit in the $O$ notation depending uniformly on compact
subsets of $\mathbb{R}^2 \times \mathbb{R} \setminus
\{(s,t) = (0,u)\}$.
\end{enumerate}
\end{lemma}

\begin{proof}
We will prove Part 2, leaving the necessary (simple) modifications for Part 1 to the reader.

We will only use the assumption that $\bm{x} = (0,0,\ldots,s,t)$ and $\bm{y} =
(0,0,\ldots,0,u)$ with $\bm{x} \neq \bm{y}$ in the last three lines of the
proof, to obtain the estimate relating $\bold{\Omega}_{d,\bm{\psi}}$ to
$\bold{\Omega}$.  Until then, $\bm{x}$ and $\bm{y}$ will denote any two points of
$\mathbb{R}^n$.  In particular, the estimate relating $\bold{C}_{d,\bm{\psi}}$
to $\bold{C}$ holds for any $\bm{x}, \bm{y} \in \mathbb{R}^n$.

We will first show that
\begin{eqnarray}\label{APPROX_UNIV2}
\bold{C}_{d,\bm{\psi}} = \bold{C}_{d,\bm{0}} + O\left(\frac{1}{d}\right).
\end{eqnarray}
Because $\bm{\psi}$ is $C^2$, 
there exists a constant $A > 0$ independent of $d$ such that
for any multi-indices ${\bm \beta} \in \mathbb{Z}_+^{n}$ and ${\bm \gamma} \in \mathbb{Z}_+^{k-n}$ we have:
\begin{align}
\left| \frac{\partial}{\partial x_i} \left(\frac{\bm x}{\sqrt{d}}\right)^{\bm \beta} \right| = O\left(\left(\frac{A}{\sqrt{d}}\right)^{|\bm \beta|} \right) \label{POWER} \\
\left| \left(\bm{\psi}\left(\frac{\bm x}{\sqrt{d}}\right)\right)^{\bm \gamma} \right| =  O\left(\left(\frac{A}{\sqrt{d}}\right)^{2|\bm \gamma|} \right) \label{PSI1}\\
\left|\frac{\partial}{\partial x_i} \left(\bm{\psi}\left(\frac{\bm x}{\sqrt{d}}\right)\right)^{\bm \gamma} \right| =  O\left(\left(\frac{A}{\sqrt{d}}\right)^{2|\bm \gamma|} \right)  \label{PSI2}
\end{align}
Both the constant $A$ and the multiplicative constants (implicit in the $O$ notation) depend
uniformly on $\bm{x}$ within compact subsets of $\mathbb{R}^n$.



We can now prove that
for any $\bm{x}, \bm{y} \in \mathbb{R}^n$
\begin{align}\label{BIG_INEQUALITY1}
\sum\limits_{\stackrel{|\bm{\beta}| + |\bm{\gamma}| \leq d,}{|\bm{\gamma}| \geq 1}} \binom{d}{\bm{\beta} \,\, \bm{\gamma}}  \,\,
 \left(\frac{\bm{x}}{\sqrt{d}}\right)^{\bm{\beta}}
\left(\bm{\psi}\left(\frac{\bm{x}}{\sqrt{d}}\right)\right)^{\bm{\gamma}} \,\,
\left(\frac{\bm{y}}{\sqrt{d}}\right)^{\bm{\beta}}
\left(\bm{\psi}\left(\frac{\bm{y}}{\sqrt{d}}\right)\right)^{\bm{\gamma}}  = O\left(\frac{1}{d}\right),
\end{align}

\begin{align}\label{BIG_INEQUALITY2}
\sum\limits_{\stackrel{|\bm{\beta}| + |\bm{\gamma}| \leq d,}{|\bm{\gamma}| \geq 1}} \binom{d}{\bm{\beta} \,\, \bm{\gamma}} \,\,
\frac{\partial}{\partial x_i} \left( \left(\frac{\bm{x}}{\sqrt{d}}\right)^{\bm{\beta}}
\left(\bm{\psi}\left(\frac{\bm{x}}{\sqrt{d}}\right)\right)^{\bm{\gamma}}\right) \,\,
\left(\frac{\bm{y}}{\sqrt{d}}\right)^{\bm{\beta}}
\left(\bm{\psi}\left(\frac{\bm{y}}{\sqrt{d}}\right)\right)^{\bm{\gamma}} = O\left(\frac{1}{d}\right), \quad \mbox{and}
\end{align}
\begin{align}\label{BIG_INEQUALITY3}
\sum\limits_{\stackrel{|\bm{\beta}| + |\bm{\gamma}| \leq d,}{|\bm{\gamma}| \geq 1}} \binom{d}{\bm{\beta} \,\, \bm{\gamma}} \,\,
\frac{\partial}{\partial x_i} \left( \left(\frac{\bm{x}}{\sqrt{d}}\right)^{\bm{\beta}}
\left(\bm{\psi}\left(\frac{\bm{x}}{\sqrt{d}}\right)\right)^{\bm{\gamma}}\right) \,\,
\frac{\partial}{\partial y_j} \left( \left(\frac{\bm{y}}{\sqrt{d}}\right)^{\bm{\beta}}
\left(\bm{\psi}\left(\frac{\bm{y}}{\sqrt{d}}\right)\right)^{\bm{\gamma}}\right) = O\left(\frac{1}{d}\right),
\end{align}
where we can have $\bm{x} = \bm{y}$ and/or $i = j$.  The constant implicit in the big-$O$ notation depends
uniformly on compact subsets of $\mathbb{R}^n \times \mathbb{R}^n$.

The proofs of each of these are essentially the same, so we'll prove (\ref{BIG_INEQUALITY3}), leaving the proofs of (\ref{BIG_INEQUALITY1}) and (\ref{BIG_INEQUALITY2}) for the reader.
We have
\begin{align}
\frac{\partial}{\partial x_i} \left( \left(\frac{\bm{x}}{\sqrt{d}}\right)^{\bm{\beta}}
\left(\bm{\psi}\left(\frac{\bm{x}}{\sqrt{d}}\right)\right)^{\bm{\gamma}}\right) \,\,
\frac{\partial}{\partial y_j} \left( \left(\frac{\bm{y}}{\sqrt{d}}\right)^{\bm{\beta}}
\left(\bm{\psi}\left(\frac{\bm{y}}{\sqrt{d}}\right)\right)^{\bm{\gamma}}\right)
 \\
=\frac{\partial}{\partial x_i} \left( \left(\frac{\bm{x}}{\sqrt{d}}\right)^{\bm{\beta}}\right)
\left(\bm{\psi}\left(\frac{\bm{x}}{\sqrt{d}}\right)\right)^{\bm{\gamma}} \,\,
\frac{\partial}{\partial y_j} \left( \left(\frac{\bm{y}}{\sqrt{d}}\right)^{\bm{\beta}}\right)
\left(\bm{\psi}\left(\frac{\bm{y}}{\sqrt{d}}\right)\right)^{\bm{\gamma}} \nonumber \\
+
\frac{\partial}{\partial x_i} \left( \left(\frac{\bm{x}}{\sqrt{d}}\right)^{\bm{\beta}}\right)
\left(\bm{\psi}\left(\frac{\bm{x}}{\sqrt{d}}\right)\right)^{\bm{\gamma}} \,\,
\left(\frac{\bm{y}}{\sqrt{d}}\right)^{\bm{\beta}}
\frac{\partial}{\partial y_j} \left(\left(\bm{\psi}\left(\frac{\bm{y}}{\sqrt{d}}\right)\right)^{\bm{\gamma}}\right) \nonumber \\
+
\left(\frac{\bm{x}}{\sqrt{d}}\right)^{\bm{\beta}}
\frac{\partial}{\partial x_i} \left(\left(\bm{\psi}\left(\frac{\bm{x}}{\sqrt{d}}\right)\right)^{\bm{\gamma}}\right) \,\,
\frac{\partial}{\partial y_j} \left( \left(\frac{\bm{y}}{\sqrt{d}}\right)^{\bm{\beta}} \right)
\left(\bm{\psi}\left(\frac{\bm{y}}{\sqrt{d}}\right)\right)^{\bm{\gamma}} \nonumber \\
+
 \left(\frac{\bm{x}}{\sqrt{d}}\right)^{\bm{\beta}}
\frac{\partial}{\partial x_i} \left(\left(\bm{\psi}\left(\frac{\bm{x}}{\sqrt{d}}\right)\right)^{\bm{\gamma}}\right) \,\,
\left(\frac{\bm{y}}{\sqrt{d}}\right)^{\bm{\beta}}
\frac{\partial}{\partial y_j} \left(\left(\bm{\psi}\left(\frac{\bm{y}}{\sqrt{d}}\right)\right)^{\bm{\gamma}}\right) \nonumber \\
\leq B \left(\frac{A^2}{d}\right)^{|\bm \beta| + 2|\bm \gamma|}. \nonumber
\end{align}
Here, $A, B > 0$ are given by (\ref{POWER}-\ref{PSI2}) and are independent of $|{\bm \beta}|, |{\bm \gamma}|$ and $d$,  but depend uniformly on $\bm{x}$ and $\bm{y}$ within compact subsets of $\mathbb{R}^n \times \mathbb{R}^n$.
In particular,
\begin{align*}
&\sum\limits_{\stackrel{|\bm{\beta}| + |\bm{\gamma}| \leq d,}{|\bm{\gamma}| \geq 1}} \binom{d}{\bm{\beta} \,\, \bm{\gamma}} \,\,
\frac{\partial}{\partial x_i} \left( \left(\frac{\bm{x}}{\sqrt{d}}\right)^{\bm{\beta}}
\left(\bm{\psi}\left(\frac{\bm{x}}{\sqrt{d}}\right)\right)^{\bm{\gamma}}\right) \,\,
\frac{\partial}{\partial y_j} \left( \left(\frac{\bm{y}}{\sqrt{d}}\right)^{\bm{\beta}}
\left(\bm{\psi}\left(\frac{\bm{y}}{\sqrt{d}}\right)\right)^{\bm{\gamma}}\right) \\
&\leq \sum\limits_{\stackrel{|\bm{\beta}| + |\bm{\gamma}| \leq d,}{|\bm{\gamma}| \geq 1}} \binom{d}{\bm{\beta} \,\, \bm{\gamma}} \,\,  B \left(\frac{A^2}{d}\right)^{|\bm \beta| + 2|\bm \gamma|} \leq \frac{B A^2}{d} \sum\limits_{\stackrel{|\bm{\beta}| + |\bm{\gamma}| \leq d,}{|\bm{\gamma}| \geq 1}} \binom{d}{\bm{\beta} \,\, \bm{\gamma}} \,\,  \left(\frac{A^2}{d}\right)^{|\bm \beta| + |\bm \gamma|}  \\
& \leq  \frac{B A^2}{d} \sum\limits_{|\bm{\beta}| + |\bm{\gamma}| \leq d} \binom{d}{\bm{\beta} \,\, \bm{\gamma}} \,\,  \left(\frac{A^2}{d}\right)^{|\bm \beta| + |\bm \gamma|} \leq \frac{B A^2}{d}  \left(1 + k \frac{A^2}{d}  \right)^d \leq \frac{C}{d}  \nonumber
\end{align*}
for some $C > 0$, since $\lim_{d\rightarrow \infty}  \left(1 + k \frac{A^2}{d}  \right)^d = {\rm e}^{k A^2}$.

The estimate (\ref{APPROX_UNIV2}) follows immediately.    For example,
{\footnotesize
\begin{eqnarray*}
E\left(\frac{\partial h_{d,\bm{\psi}}(\bm{x})}{\partial x_i} \frac{\partial h_{d,\bm{\psi}}(\bm{y})}{\partial y_j}\right) = \hspace{5in} \\
\sum\limits_{|\bm{\beta}|  \leq d} \binom{d}{\bm{\beta} } \,\,
\frac{\partial}{\partial x_i} \left( \frac{\bm{x}}{\sqrt{d}}\right)^{\bm{\beta}}
\frac{\partial}{\partial y_j}  \left(\frac{\bm{y}}{\sqrt{d}}\right)^{\bm{\beta}} + 
\sum\limits_{\stackrel{|\bm{\beta}| + |\bm{\gamma}| \leq d,}{|\bm{\gamma}| \geq 1}} \binom{d}{\bm{\beta} \,\, \bm{\gamma}} \,\,
\frac{\partial}{\partial x_i} \left( \left(\frac{\bm{x}}{\sqrt{d}}\right)^{\bm{\beta}}
\left(\bm{\psi}\left(\frac{\bm{x}}{\sqrt{d}}\right)\right)^{\bm{\gamma}}\right) \,\,
\frac{\partial}{\partial y_j} \left( \left(\frac{\bm{y}}{\sqrt{d}}\right)^{\bm{\beta}}
\left(\bm{\psi}\left(\frac{\bm{y}}{\sqrt{d}}\right)\right)^{\bm{\gamma}}\right)  \\
= E\left(\frac{\partial h_{d,\bm{0}}(\bm{x})}{\partial x_i} \frac{\partial h_{d,\bm{0}}(\bm{y})}{\partial y_j}\right) + O\left(\frac{1}{d}\right).
\end{eqnarray*}
}

We will now show that
\begin{eqnarray}\label{EQN_FLAT_CONV}
\bold{C}_{d,\bm{0}} = \bold{C} + O\left(\frac{1}{d}\right).
\end{eqnarray}
It follows from a calculation
analogous to that from Lemma \ref{COV_MATRIX_CALC}, but with a rescaling by $1/\sqrt{d}$ and the fact that
\begin{eqnarray*}
\left(1+\frac{x}{d}\right)^d = e^x + O\left(\frac{1}{d}\right),
\end{eqnarray*}
with the constant depending uniformly on $x \in \mathbb{R}$.  Rather than including the computation for each of the eight different types
of expectations, we simply list one of the more complicated ones here:

\begin{align*}
E\left( \frac{\partial h_{d,\bm{0}}(\bm{x})}{\partial x_i} \frac{\partial
h_{d,\bm{0}}(\bm{y})}{\partial y_i} \right) &= 
\left(1 + \frac{\bm{x}\cdot \bm{y}}{d} + \frac{x_i y_i (d-1)}{d} \right)
\left( 1 + \frac{\bm{x} \cdot \bm{y}}{d} \right)^{d-2} \quad \mbox{while}  \\ 
E\left(\frac{\partial
f(\bm{x})}{\partial x_i }\frac{\partial f(\bm{y})}{\partial y_i }\right) &=
(1+x_i y_i) e^{{\bm x} \cdot {\bm y}}.
\end{align*}
\vspace{0.1in}
Combining (\ref{APPROX_UNIV2}) and (\ref{EQN_FLAT_CONV}) we conclude that
$\bold{C}_{d,\bm{\psi}} = \bold{C} + O\left(\frac{1}{d}\right)$.

If $\bm{x} = (0,\ldots,0,s,t)$ and $\bm{y} = (0,\ldots,0,u)$, with $\bm{x} \neq
\bm{y}$, Part 2 of Lemma \ref{GAF_POS_DEF} gives that $\bold{C}$ is positive
definite.  Then, there is a neighborhood of $\bold{C}$ in the space of all
$2n^2 \times 2n^2 $ matrices on which taking the inverse is a differentiable
map.  Therefore, $\bold{C}_{d,\bm{\psi}}^{-1} = \bold{C}^{-1} +
O\left(\frac{1}{d}\right)$ and hence $\bold{\Omega}_{{d},\bm{\psi}} =
\bold{\Omega} + O\left(\frac{1}{d}\right)$. 
\end{proof}

\subsection{Proof of the local universality theorem}

{\bf Proof of Theorem \ref{LOCAL_UNIVERSALITY}.}
Throughout the proof we will use that we have normalized so that $\bm{x} = (0,\ldots,0,s,t)$ and $\bm{y} = (0,\ldots,0,u)$ with $\bm{x} \neq \bm{y}$. 

Part 2 of Lemma \ref{GAF_POS_DEF} gives that the covariance matrix
$\bold{C}$ for random vector (\ref{Vector}) applied to ensemble (\ref{GAF}) is
positive definite.  Therefore, Part 2 of Lemma \ref{LEM_CONV_COV_MATRICES} gives that covariance matrix $\bold{C}_{d,\bm{\psi}}$ for
random vector (\ref{Vector}) applied to ensemble~(\ref{TANGENT_DIST}) will also
be positive definite for $d$ sufficiently large.
We can therefore apply the Kac-Rice formula (\ref{K1}) to
ensemble (\ref{TANGENT_DIST}) obtaining  
\begin{equation}\label{EQN_CRAZY_KAC}
K_{d,{\bm \psi}}\left(\bm{x},\bm{y}\right)=\frac{1}{\left(2\pi\right)^{n\left(n+1\right)} \rho_{d,{\bm \psi}}(\bm{x})\rho_{d,{\bm \psi}}(\bm{y})\sqrt{\det \bold{C}_{d,{\bm \psi}}}} \int\limits_{\mathbb{R}^{2n^2}} | \det \bm{\xi}| |\det \bm{\eta}| \mathrm{e}^{-\frac{1}{2} \left(\bm{\Omega}_{d,{\bm \psi}}\bm{u},\bm{u}\right)} d\bm{u}. 
\end{equation}
We will show that (\ref{EQN_CRAZY_KAC}) differs by
$O\left(\frac{1}{\sqrt{d}}\right)$ from the result obtained when applying the Kac-Rice
Formula (\ref{K1}) to the $\isoR$-invariant ensemble~(\ref{GAF}).

Let us first consider the prefactor from (\ref{EQN_CRAZY_KAC}).
To show that $\rho_{d,{\bm \psi}}(\bm{x}) = \rho(\bm{x}) +  O\left(\frac{1}{\sqrt{d}}\right)$, we apply
Lemma \ref{COCV3} to the Kac-Rice formula for density (\ref{K_density}).  
This follows
because the matrix $\bold{\Omega}$ in (\ref{K_density}) is positive definite, by Part 1 of Lemma \ref{GAF_POS_DEF},
and because 
$\bold{\Omega}_{{d},{\bm{\psi}}} = \bold{\Omega} + O\left(\frac{1}{d}\right)$, by Part 1 of Lemma \ref{LEM_CONV_COV_MATRICES}.
The same estimate holds at $\bm{y}$.
Meanwhile, since $\bold{C}_{{d},{\bm \psi}} = \bold{C} + O\left(\frac{1}{d}\right)$, with $\bold{C}$ positive definite, it follows immediately
that
\begin{eqnarray*}
\frac{1}{\sqrt{\det \bold{C_{d}}}} = \frac{1}{\sqrt{\det \bold{C}}}  + O\left(\frac{1}{d}\right).
\end{eqnarray*}

We now consider the integral in (\ref{EQN_CRAZY_KAC}).  Part 2 of Lemma
\ref{GAF_POS_DEF} gives that the matrix $\bold{\Omega}$ used when applying the Kac-Rice formula~(\ref{K1}) to the $\isoR$-invariant ensemble (\ref{GAF}) is positive definite.   Meanwhile, Part 2 of
Lemma \ref{LEM_CONV_COV_MATRICES} gives that $\bold{\Omega}_{d,{\bm \psi}} =
\bold{\Omega} + O\left(\frac{1}{d}\right)$.  Thus, Lemma \ref{COCV3} gives that
the integral from (\ref{EQN_CRAZY_KAC}) differs by
$O\left(\frac{1}{\sqrt{d}}\right)$ from the integral in (\ref{K1}), when
(\ref{K1}) is applied to (\ref{GAF}).
Each of the lemmas used asserts that the constants depend uniformly on compact
subsets of $\mathbb{R}^2 \times \mathbb{R} \setminus \{(s,t) = (0,u)\}$. 
 \qed

\vspace{0.1in}
\noindent
Since we reduced the statement of Theorem \ref{THM:UNIVERSALITY} to Theorem
\ref{LOCAL_UNIVERSALITY} in Subsection \ref{SUBSEC_REDUCTION}, we have also
finished the proof of Theorem~\ref{THM:UNIVERSALITY}.

\section{Finite degree restrictions to submanifolds $M \subset \mathbb{RP}^k$ depend on the geometry}\label{SEC:GEOMETRY_FINITE_D}

We present a simple example illustrating that the constant from the leading term in the correlation
function can depend on the geometry of a submanifolds $M \subset \mathbb{RP}^k$ if the degree $d$ of the ensemble is finite.
Thus, it is not possible to prove a universal formula for the short-range asymptotics at finite degree.

We consider the ensemble $SO(3)$-invariant polynomials of degree $3$ because for degrees $1$ and $2$ the covariance matrix for random vector (\ref{Vector}) is not positive definite. Restricted to a system of affine coordinates $(x,y) \rightarrow [x:y:1]$
each such polynomial has the form
\begin{eqnarray*}
F_3(x,y) = a_{0,0} + \sqrt{3} a_{1,0} x + \sqrt{3} a_{0,1} y + \sqrt{6} a_{1,1} x y + \sqrt{3} a_{2,0} x^2 + \sqrt{3} a_{0,2} y^2  \\ + \sqrt{3} a_{2,1} x^2 y + \sqrt{3} a_{1,2} x y^2 + a_{3,0} x^3 + a_{0,3} y^3,
\end{eqnarray*}
where each of the coefficients is chosen iid with respect to the standard normal distribution, $\mathcal{N}\left(0, 1\right)$.

Let $M \subset \mathbb{R}^2$ be given by $y = b x^2$.  As in the previous section, we parameterize $M$ by the $x$ coordinate, in this case forming the one-variable ensemble of random polynomials that depend on $b$ as a parameter:
\begin{eqnarray}\label{EQN:ENSEMBLE_H_CURVED}
H_{3,a} (x) = a_{0,0} + \sqrt{3} a_{1,0} x + \sqrt{3} a_{0,1} b x^2 + \sqrt{6} a_{1,1} b x^3 + \sqrt{3} a_{2,0} x^2 + \sqrt{3} a_{0,2} b^2 x^4 \\ + \sqrt{3} a_{2,1} b x^4 + \sqrt{3} a_{1,2} b^2 x^5 + a_{3,0} x^3 + a_{0,3} b^3 x^6, \nonumber
\end{eqnarray}
In order for our results from Sections \ref{SEC:COVARIANCE} and \ref{SEC:SHORT_RANGE} to apply here, 
we multiply by the prefactor: $h_{3,b} := \left(1+\frac{x^2}{2}\right)^3 H_{3,b} (x)$.

We will use the Kac-Rice formula (\ref{K1}) to show that the value of $b$ affects the constant term in the short-range asymptotics for the correlation function $K(x,y)$ with $x = -\frac{t}{2}$ and $y = \frac{t}{2}$. 

An easy calculation using the Kac-Rice formula (\ref{K_density}) for density of zeros gives that
\begin{eqnarray*}
\rho(x) = \rho(y) = \frac{\sqrt{3}}{\pi} + O(t^2),
\end{eqnarray*}
where the constant term agrees with the result (\ref{Density_FINITE_D}) that is stated in the introduction for the $SO(2)$-invariant polynomials of degree $3$.

The covariance matrix for random vector (\ref{Vector}) applied to the ensemble $h_{3,b}$ with $x = -\frac{t}{2}$ and $y = \frac{t}{2}$ is
\begin{align}
\bold{C} = \left[ \begin {array}{cccc} \alpha&\delta&\mu&-\nu
\\ \noalign{\medskip}\delta&\gamma&\nu&\tau\\ \noalign{\medskip}\mu&
\nu&\alpha&-\delta\\ \noalign{\medskip}-\nu&\tau&-\delta&\gamma
\end {array} \right] 
\end{align}
where 
\begin{align*}
\alpha = {\frac { \left( {k}^{2}{t}^{4}+4{t}^{2}+16 \right) ^{3}}{4096}},  \quad
\delta =  -{\frac {3t \left( {k}^{2}{t}^{2}+2 \right)  \left( {k}^{2}{t}^{4}+4
{t}^{2}+16 \right) ^{2}}{1024}}, \quad \mu =  {\frac { \left( {k}^{2}{t}^{4}-4{t}^{2}+16 \right) ^{3}}{4096}},  \hspace{1in} \\
\gamma =  {\frac { \left( 3{k}^{2}{t}^{4}+12{t}^{2}+48 \right)  \left( 3{k
}^{4}{t}^{6}+13{k}^{2}{t}^{4}+16{k}^{2}{t}^{2}+12{t}^{2}+16
 \right) }{256}}, \quad
\nu = -{\frac {3t \left( {k}^{2}{t}^{2}-2 \right)  \left( {k}^{2}{t}^{4}-4
{t}^{2}+16 \right) ^{2}}{1024}}
, \\
 \mbox{and} \qquad \tau = -{\frac { \left( 3{k}^{2}{t}^{4}-12{t}^{2}+48 \right)  \left( 3{
k}^{4}{t}^{6}-13{k}^{2}{t}^{4}+16{k}^{2}{t}^{2}+12{t}^{2}-16
 \right) }{256}}. 
\end{align*}
This matrix has exactly the same structure as that from Sections \ref{SEC:COVARIANCE} and \ref{SEC:SHORT_RANGE}, in the case that the dimension is $1$.  
Therefore, the submatrix 
$\bold{\Omega}$ of $C^{-1}$ from the Kac-Rice formula (\ref{K1}) is diagonalized in precisely the same fashion, with the eigenvalues
satisfying
\begin{eqnarray*}
\lambda_3^{-1/2} = \sqrt{6(1+b^2)} \,\, t + O(t^2) \qquad \mbox{and} \qquad \lambda_4^{-1/2} =  \sqrt{\frac{1}{2} + 3 b^2} \,\, t^2 + O(t^3).
\end{eqnarray*}
(We call them $\lambda_3$ and $\lambda_4$ in order to be consistent with the previous sections.)
We also have
\begin{eqnarray*}
\det{C} =  \left( 54\,{b}^{4}+63\,{b}^{2}+9 \right) {t}^{8}+O \left( {t}^{10} \right).
\end{eqnarray*}

The calculation of the short-range asymptotics done in Section \ref{SEC:SHORT_RANGE} applies here, with the minor modifications to $\lambda_3^{-1/2}, \lambda_4^{-1/2},$ and $\det{C}$ listed above.  One obtains

\begin{proposition}
The correlation between zeros for ensemble (\ref{EQN:ENSEMBLE_H_CURVED}) satisfies 
\begin{eqnarray*}
K\left(-\frac{t}{2},\frac{t}{2} \right) = \frac{\pi}{2 \sqrt{3}} (1+b^2) \ t + O(t^2).
\end{eqnarray*}
In particular, the leading term depends on the curvature of $M$ at $0$.
\end{proposition}

\begin{question}
In the general setting of $M \subset \mathbb{RP}^k$ how does the constant in the leading order asymptotics near $p \in M$ depend on the local geometry
of $M$ at $p$?
\end{question}

\section{Universality in the complex setting}
\label{SEC:COMPLEX_CASE}

We begin by adapting the Kac-Rice formulae (\ref{K_density})
and (\ref{K1}) to the complex setting.  As the modifications
are nearly identical, we will discuss the formula for correlations, leaving the
formula for density to the reader.

Suppose that $\bm{h} = (h_1,h_2,\ldots,h_n) : \mathbb{C}^n \rightarrow \mathbb{C}^n$ is a Gaussian random analytic
function with complex Gaussian coefficients.  
Let $\begin{displaystyle}{\bm \xi}\end{displaystyle}$ and ${\bm \eta}$ be the $n\times n$ complex matrices whose rows are ${\bm \xi}_{1}, \dots {\bm \xi}_{n}$ and ${\bm \eta}_{1}, \dots {\bm \eta}_{n}$, respectively. Let $\begin{displaystyle}\bm{u} =\left[ \begin {array}{ccccccc} { \bm{\xi}_{1}}&{ \bm{\eta}_{1}}&{ \bm{\xi}_{2}}&{ \bm{\eta}_{2}}&\dots&{ \bm{\xi}_{n}}&{ \bm{\eta}_{n}}\end {array} \right]^{\intercal} \end{displaystyle}$, the vector formed by alternating the vectors $\bm{\xi}_{i}$ and~$\bm{\eta}_{i}$.

\begin{proposition} \label{ComplexKacRiceFormula} Suppose that the covariance matrix $C = E(v_i \overline{v_j})$ of the random vector (\ref{Vector}) is positive definite.  Then, the two-point correlation function for the zeroes of the system $\bm{h}$ is:
\begin{equation} \label{K1COMPLEX}
K_n\left(\bm{x},\bm{y}\right)=\frac{1}{\pi^{n\left(n+1\right)}
\rho(\bm{x})\rho(\bm{y})\det \bold{C}}
\int\limits_{\mathbb{C}^{2n^2}}  \det (\bm{\xi}^* \bm{\xi}) \det
(\bm{\eta}^* \bm{\eta}) \mathrm{e}^{- \left(\bm{\Omega}\bm{u},\bm{u}\right)} d\bm{u} d\overline{\bm{u}},
\end{equation}
where $\bm{\Omega}$ is the matrix of the elements of $\bold{C}^{-1}$
left after removing the rows and columns 
with indices congruent to 1 modulo $n+1$, $*$ denotes conjugate transpose, and $(\, , \,)$ denotes the Hermitian inner
product.
\end{proposition}

\noindent
Our General Kac-Rice Formula (Proposition \ref{KAC_RICE_MULTI_DENSITY}) can be easily adapted to the complex setting,
from which Proposition~\ref{ComplexKacRiceFormula} follows easily.
Alternatively, Proposition \ref{ComplexKacRiceFormula}
follows from \cite[Theorem 2.1]{BSZUniversality} by using the suitable
Gaussian density $D_k(0,\xi,z)$ in their formula 32 and normalizing by the
density at the two points $\bm{x}$ and $\bm{y}$. 

With these modifications to the Kac-Rice formulae, the proof of
Theorem~\ref{THM:UNIVERSALITYC} adapts nearly verbatim from the proof of
Theorem~\ref{THM:UNIVERSALITY}.  We list below the simple modifications that need
to be checked:
\begin{enumerate}
\item  The proof of Lemma \ref{COCV3} adapts easily to the integral expression in (\ref{K1COMPLEX}) and to the analogous formula for the densities $\rho(\bm{x})$ and $\rho(\bm{y})$.

\item  The proof of Lemma~\ref{GAF_POS_DEF} is easily adapted.  More specifically, the determinant of the $6 \times 6$ block
analogous to (\ref{WEIRD_DET}) equals:
\begin{eqnarray*}
{{\rm e}^{6\,{\rm Re} \left( t\overline{u} \right) }} \left( {{\rm e}^{
 \left(  \left| s \right|  \right) ^{2}+ \left(  \left| t-u \right| 
 \right) ^{2}}}-{{\rm e}^{2\, \left(  \left| s \right|  \right) ^{2}+2
\, \left(  \left| t-u \right|  \right) ^{2}}} \right)  \left(  \left( 
 \left(  \left| s \right|  \right) ^{2}+ \left(  \left| t-u \right| 
 \right) ^{2} \right) ^{2}+3 \right) +{{\rm e}^{3\, \left(  \left| s
 \right|  \right) ^{2}+3\, \left(  \left| t-u \right|  \right) ^{2}}}- 1,
\end{eqnarray*}
which is positive for $(s,t) \neq (0,u)$.

\item  The proof of Lemma \ref{LEM_CONV_COV_MATRICES} applies after verifying that, when expressed in local coordinates,
the covariance matrix for the rescaled $SU(n+1)$-invariant polynomials differs from that of the $\isoC$-invariant ensemble by $O\left(\frac{1}{d}\right)$.  
(These covariances are listed in \cite[Sections 2.4 and 4]{BSZUniversality}.)

\end{enumerate}

\appendix

\section{Proof of Lemma \ref{COCV3}}\label{SEC:Lemma8}

We will need the following lemma to prove Lemma \ref{COCV3}.
\begin{lemma} \label{Tails}  
For any positive definite $mn^2 \times mn^2$ matrix $\bold{A}$, there exists a constant $D>0$ such that, for $R$ sufficiently large, 
\begin{equation}  \label{TAILS_DENSITY}
\left| \int\limits_{||{\bm u}|| > R} \prod_{i=1}^m  | \det \bm{\xi}^i|  \mathrm{e}^{-\frac{1}{2} (\bold{A} {\bm u}, \bm{u})} d\bm{u}\right| \leq  \mathrm{e}^{-D R}.
\end{equation}
(Here $\bm{u}$ is as in (\ref{EQN_DEF_U}).
\end{lemma}

\begin{proof}
First, we note that $(\bold{A}\bm{u}, \bm{u})\geq \lambda_{\rm min} ||\bm{u}||^2$. The left side of inequality (\ref{TAILS_DENSITY}) is bounded by
\begin{equation}
T := \left| \int\limits_{||\bm{u}|| > R} \prod_{i=1}^m  | \det \bm{\xi}^i| \mathrm{e}^{-\frac{1}{2} \lambda_{\rm min} ||\bm{u}||^2} d\bm{u}\right| \leq \mathrm{Vol}(\mathbb{S}^{mn^2-1}) n!^m \int\limits_{R}^{\infty} r^{2 m n^2 - 1}\mathrm{e}^{-\frac{\lambda_{\rm min}}{2} r^2} dr,
\end{equation}
using that $\prod_{i=1}^m  | \det \bm{\xi}^i| \leq n!^m ||\bm{u}||^{mn^2}$. Let
$a := 2m n^2  - 1$ and $b :=\frac{\lambda_{\rm min}}{2}$.

\begin{equation}
\int\limits_{R}^{\infty} r^{a}\mathrm{e}^{-b r^2} dr \leq \int\limits_{R}^{\infty} \mathrm{e}^{-b r^2 + a r} dr = \frac{\sqrt{\pi} \mathrm{e}^{\frac{a^2}{4b}} }{2} \mathrm{erfc}\left(\sqrt{b}R - \frac{a}{2\sqrt{b}} \right).
\end{equation}
The result then follows because $\begin{displaystyle}\mathrm{erfc}(x) := \frac{2}{\sqrt{\pi}}\int\limits_{x}^{\infty} \mathrm{e}^{-t^2} dt \leq e^{-x^2} \end{displaystyle}$.
\end{proof}

\noindent
{\bf Proof of Lemma \ref{COCV3}.}
We first split the left side of the inequality into integrals with $||\bm{u}|| < R$ and $||\bm{u}|| > R$ for some~$R$:
\begin{align} \label{SeparateTails}  & \left| \int\limits_{\mathbb{R}^{mn^2}}  \prod_{i=1}^m  | \det \bm{\xi}^i| \mathrm{e}^{-\frac{1}{2} \left(\bold{B}\bm{u},\bm{u}\right)}  -  \int\limits_{\mathbb{R}^{mn^2}}  \prod_{i=1}^m  | \det \bm{\xi}^i| \mathrm{e}^{-\frac{1}{2} \left(\bold{A}\bm{u},\bm{u}\right)} d\bm{u} \right| \\
 &\leq   \left| \int\limits_{||\bm{u}|| < R}  \prod_{i=1}^m  | \det \bm{\xi}^i| \mathrm{e}^{- \frac{1}{2} (\bold{A}\bm{u},\bm{u})} \left(\mathrm{e}^{-\frac{1}{2} \left((\bold{B}-\bold{A}) \bm{u},\bm{u}\right)}-1\right) d\bm{u}\right| \\ 
&+ \left| \int\limits_{||\bm{u}|| > R}  \prod_{i=1}^m  | \det \bm{\xi}^i| \mathrm{e}^{-\frac{1}{2} (\bold{B} \bm{u}, \bm{u})} d\bm{u}\right| +  \left| \int\limits_{||\bm{u}|| > R} \prod_{i=1}^m  | \det \bm{\xi}^i| \mathrm{e}^{-\frac{1}{2} (\bold{A} \bm{u}, \bm{u})} d\bm{u}\right| . \end{align}
\noindent
For $\bold{B}$ sufficiently close to $\bold{A}$, {\rm min} $(\bold{B}\bm{u},\bm{u})\sqrt{2} \geq (\bold{A}\bm{u},\bm{u}) $. Thus, Lemma \ref{Tails} gives that the two latter summands of (\ref{SeparateTails}) are both bound by $\mathrm{e}^{-D R}$.

We use H\"older's Inequality to bound the first summand in (\ref{SeparateTails}) with
\begin{equation} \ \label{Holders}
 \left| \int\limits_{||\bm{u}|| < R}  \prod_{i=1}^m  | \det \bm{\xi}^i| \mathrm{e}^{- \frac{1}{2} (\bold{A}\bm{u},\bm{u})} d\bm{u} \right| \left|\left| \mathrm{e}^{-\frac{1}{2} \left((\bold{B}-\bold{A}) \bm{u},\bm{u}\right)}-1\right|\right|_{L^{\infty} (||\bm{u}|| < R)} = D_5 \left|\left| \mathrm{e}^{-\frac{1}{2} \left((\bold{B}-\bold{A}) \bm{u},\bm{u}\right)}-1\right|\right|_{L^{\infty} (||\bm{u}|| < R)}.
\end{equation}

Let $\bold{L} = \bold{A} - \bold{B}$. Since $|| \mathrm{e}^{\frac{1}{2} x} - 1 || < x$ as $x$ approaches $1$,
\begin{align}
\left|\left| \mathrm{e}^{\frac{1}{2} \left(\bold{L} \bm{u},\bm{u}\right)}-1 \right|\right|_{L^{\infty} (||\bm{u}|| < R)} \leq \left|\left| {\left(\bold{L} \bm{u},\bm{u}\right)} \right|\right|_{L^{\infty} (||\bm{u}|| < R)}\leq mn^2\left|\left|\bold{L}\right|\right|_{\infty} ||\bm{u}||^2_{||\bm{u}||<R} \leq   mn^2 R^2\left|\left|\bm{L}\right|\right|_{\infty}  .
\end{align}

Therefore, we have
\begin{equation} \label{FinalREquation}   \left| \int\limits_{\mathbb{R}^{2n^2}}| \det \bm{\xi}| |\det \bm{\eta}| \mathrm{e}^{- \frac{1}{2} (\bold{A}\bm{u},\bm{u})} \left(\mathrm{e}^{-\frac{1}{2} \left((\bold{B}-\bold{A}) \bm{u},\bm{u}\right)}-1\right) d\bm{u}\right| \leq  mn^2 D_5 R^2 \left|\left|\bold{A-B}\right|\right|_{\infty} + 2 \mathrm{e}^{-D R}
\end{equation}
The result follows if we set $R = ||\bold{A}-\bold{B}||_{\infty}^{-1/4}$.
\qed

\section{Moments of the Volume of a Random Unit Parallelotope} \label{kthMomentProof}
The derivation of the formula for the density (\ref{Density}) and the short-range asymptotics from Theorems \ref{CORR_FINITE_D} and \ref{CORR_ISOM_INV} require the following formula:
\begin{proposition} \label{kthMoment}
Consider $n$ random unit vectors in $\mathbb{R}^n$ chosen independently with respect to spherical measure. The $k$th moment of the volume $V$ of the parallelotope formed by these vectors is \begin{displaymath} E[V^k]  =  \left(\frac{\Gamma\left(\frac{n}{2}  \right)}{\Gamma\left(\frac{n+k}{2}  \right)}\right)^{n-1} \prod\limits_{i=1}^k \frac{\Gamma\left(\frac{n-1+i}{2} \right)}{\Gamma\left( \frac{i}{2}\right)}. \end{displaymath}
\end{proposition}
This result is proved in greater generality in \cite{Miles} (see Equation 23); however, we will give a simple derivation of the formula.

\begin{proof}
The method for finding this moment involves fixing each of the vectors $v_i$ that determine the parallelotope one at a time with respect to the parallelotope in $\mathbb{R}^{i-1}$ described by the previous $i-1$ vectors, weighting each newly added vector based on the probability of obtaining a given height off of the previous $\left(i-1\right)-$parallelotope. 

Thus, after fixing the first vector, we see that the height of the second vector off of the first vector in terms of the angle $\theta$ off of the normal to the first is $\cos\left(\theta\right)$ and the probability density of obtaining this height is \begin{displaymath}\frac{\sin^{n-2}\theta}{\int\limits_0^{\frac{\pi}{2}}\sin^{n-2}\theta \ d\theta}.\end{displaymath}
Yet, from the next vector on, the vector can vary in two directions, along two different spheres, one of dimension $i-1$ and the other of dimension $n-i-1$, in order to maintain the same height. Thus, for the $i$th vector, we have that the probability density of obtaining each height $\cos \theta$ for an angle $\theta$ off the normal vector of the base given by the first $i-1$ vectors is \begin{displaymath} \frac{\sin^{n-i-1}\theta\cos^{i-1} \theta}{\int\limits_0^{\frac{\pi}{2}}\sin^{n-i-1}\theta \cos^{i-1} \theta \ d \theta}.\end{displaymath}
To express the $k$th power of the volume, we multiply the $k$th power of each of the heights together.
Thus,
\begin{align} \label{kthMomentParalletope} E[V^k]  = \prod\limits_{i=1}^{n-1} \left(\frac{\int\limits_{0}^{\pi/2} \sin^{n-i-1} \theta \cos^{i-1+k} \theta d\theta}{\int\limits_{0}^{\pi/2} \sin^{n-i-1} \theta \cos^{i-1} \theta d\theta}\right) &= \left(\frac{\Gamma\left(\frac{n}{2}  \right)}{\Gamma\left(\frac{n+k}{2}  \right)}\right)^{n-1} \prod\limits_{i=1}^{n-1} \left( \frac{\Gamma\left(\frac{n-i}{2}  \right) \Gamma\left(\frac{i+k}{2} \right)}{\Gamma\left(\frac{n-i}{2}\right)\Gamma\left(\frac{i}{2}\right)} \right)
\\ &= \left(\frac{\Gamma\left(\frac{n}{2}  \right)}{\Gamma\left(\frac{n+k}{2}  \right)}\right)^{n-1} \prod\limits_{i=1}^k \frac{\Gamma\left(\frac{n-1+i}{2} \right)}{\Gamma\left( \frac{i}{2}\right)} .
\end{align}
\end{proof}

\noindent
{\bf Acknowledgements:}
This work was supported in part by the National Science Foundation 
[DMS-1265172 to P.B, DMS-1102597 and DMS-1348589 to R.K.W.R.].

The key technique for computing short-range asymptotics for
$K(\bm{x},\bm{y})$ in Section \ref{SEC:SHORT_RANGE} was the 2013 Siemens competition and 2014 Intel Science Talent Search
project of Homma, which he developed while a junior and senior at Carmel High School.

We thank Liviu Nicolaescu for his interesting comments.  We also thank the
referee for his or her very careful reading of the paper and suggestions that have improved the exposition.
\bibliographystyle{plain}
\singlespacing
{\bibliography{correlations}}

\begin{thebibliography}{10}

\bibitem{MAPLE}
Maple computer algebra system.
\newblock \url{http://www.maplesoft.com}.


\bibitem{AT}
Robert~J. Adler and Jonathan~E. Taylor.
\newblock {\em Random fields and geometry}.
\newblock Springer Monographs in Mathematics. Springer, New York, 2007.

\bibitem{AW}
Jean-Marc Aza{\"{\i}}s and Mario Wschebor.
\newblock {\em Level sets and extrema of random processes and fields}.
\newblock John Wiley \& Sons, Inc., Hoboken, NJ, 2009.


\bibitem{BleherDi1}
Pavel Bleher and Xiaojun Di.
\newblock Correlations between zeros of a random polynomial.
\newblock {\em J. Statist. Phys.}, 88(1-2):269--305, 1997.

\bibitem{BleherDi2}
Pavel Bleher and Xiaojun Di.
\newblock Correlations between zeros of non-{G}aussian random polynomials.
\newblock {\em Int. Math. Res. Not.}, (46):2443--2484, 2004.

\bibitem{BleherRidzal}
Pavel Bleher and Denis Ridzal.
\newblock {${\rm SU}(1,1)$} random polynomials.
\newblock {\em J. Statist. Phys.}, 106(1-2):147--171, 2002.

\bibitem{BSZUniversality2}
Pavel Bleher, Bernard Shiffman, and Steve Zelditch.
\newblock Poincar\'e-{L}elong approach to universality and scaling of
  correlations between zeros.
\newblock {\em Comm. Math. Phys.}, 208(3):771--785, 2000.

\bibitem{BSZUniversality}
Pavel Bleher, Bernard Shiffman, and Steve Zelditch.
\newblock Universality and scaling of correlations between zeros on complex
  manifolds.
\newblock {\em Invent. Math.}, 142(2):351--395, 2000.

\bibitem{BSZ3}
Pavel Bleher, Bernard Shiffman, and Steve Zelditch.
\newblock Correlations between zeros and supersymmetry.
\newblock {\em Comm. Math. Phys.}, 224(1):255--269, 2001.
\newblock Dedicated to Joel L. Lebowitz.

\bibitem{Bogomolny1}
E.~Bogomolny, O.~Bohigas, and P.~Leboeuf.
\newblock Quantum chaotic dynamics and random polynomials.
\newblock {\em J. Statist. Phys.}, 85(5-6):639--679, 1996.

\bibitem{Burgisser}
Peter B{\"u}rgisser.
\newblock Average {E}uler characteristic of random real algebraic varieties.
\newblock {\em C. R. Math. Acad. Sci. Paris}, 345(9):507--512, 2007.

\bibitem{Kostlan}
Alan Edelman and Eric Kostlan.
\newblock How many zeros of a random polynomial are real?
\newblock {\em Bull. Amer. Math. Soc. (N.S.)}, 32(1):1--37, 1995.

\bibitem{FLL}
Y.~Fyodorov, A.~Lerario, and Lundberg E.
\newblock On the number of connected components of random algebraic
  hypersurfaces.
\newblock Preprint: \url{http://arxiv.org/abs/1404.5349}.

\bibitem{GW1}
Damien Gayet and Jean-Yves Welschinger.
\newblock Exponential rarefaction of real curves with many components.
\newblock {\em Publ. Math. Inst. Hautes \'Etudes Sci.}, (113):69--96, 2011.

\bibitem{GW2}
Damien Gayet and Jean-Yves Welschinger.
\newblock What is the total {B}etti number of a random real hypersurface?
\newblock {\em J. Reine Angew. Math.}, 689:137--168, 2014.


\bibitem{GW3}
Damien Gayet and Jean-Yves Welschinger.
\newblock {B}etti numbers of random real hypersurfaces and determinants of
  random symmetric matrices.
\newblock To appear in J. Eur. Math. Soc.,
  \url{http://arxiv.org/abs/1207.1579}.

\bibitem{GW4}
Damien Gayet and Jean-Yves Welschinger.
\newblock Lower estimates for the expected {B}etti numbers of random real
  hypersurfaces.
\newblock {\em J. Lond. Math. Soc. (2)}, 90(1):105--120, 2014.

\bibitem{GW5}
Damien Gayet and Jean-Yves Welschinger.
\newblock Expected topology of random real algebraic submanifolds.
\newblock To appear in J. Inst. Math. Jussieu,
  \url{http://arxiv.org/abs/1307.5287}.

\bibitem{Hida}
Takeyuki Hida and Masuyuki Hitsuda.
\newblock {\em Gaussian processes}, volume 120 of {\em Translations of
  Mathematical Monographs}.
\newblock American Mathematical Society, Providence, RI, 1993.
\newblock Translated from the 1976 Japanese original by the authors.


\bibitem{Peres}
J.~Ben Hough, Manjunath Krishnapur, Yuval Peres, and B{\'a}lint Vir{\'a}g.
\newblock {\em Zeros of {G}aussian analytic functions and determinantal point
  processes}, volume~51 of {\em University Lecture Series}.
\newblock American Mathematical Society, Providence, RI, 2009.

\bibitem{IP}
I.~A. Ibragimov and S.~S. Podkorytov.
\newblock On random real algebraic surfaces.
\newblock {\em Dokl. Akad. Nauk}, 343(6):734--736, 1995.

\bibitem{Zaporozhets}
I.~A. Ibragimov and D.~N. Zaporozhets.
\newblock On the area of a random surface.
\newblock {\em Zap. Nauchn. Sem. S.-Peterburg. Otdel. Mat. Inst. Steklov.
  (POMI)}, 384(Veroyatnost i Statistika. 16):154--175, 312, 2010.

\bibitem{Kac}
M.~Kac.
\newblock On the average number of real roots of a random algebraic equation.
\newblock {\em Bull. Amer. Math. Soc.}, 49:314--320, 1943.

\bibitem{KOSTLAN1}
E.~Kostlan.
\newblock On the distribution of roots of random polynomials.
\newblock In {\em From {T}opology to {C}omputation: {P}roceedings of the
  {S}malefest ({B}erkeley, {CA}, 1990)}, pages 419--431. Springer, New York,
  1993.

\bibitem{LerarioLundberg}
Antonio Lerario and Erik Lundberg.
\newblock Statistics on {H}ilbert's {S}ixteenth {P}roblem.
\newblock To appear in {\it International Math Research Notices}; see also
  \url{http://arxiv.org/abs/1212.3823v2}.

\bibitem{MO}
MathOverflow Website Discussion.  \\
\newblock \url{http://mathoverflow.net/questions/130496/continuous-dependence-of-the-expectation-of-a-r-v-on-the-probability-measure}

\bibitem{Miles}
R.~E. Miles.
\newblock Isotropic random simplices.
\newblock {\em Advances in Appl. Probability}, 3:353--382, 1971.

\bibitem{Nastasescu}
Maria Nastasescu.
\newblock The number of ovals of a random real plane curve.
\newblock Preprint:
  \url{http://www.its.caltech.edu/~mnastase/Maria_Nastasescu_Thesis.pdf}.

\bibitem{Nicolaescu}
Liviu I. Nicolaescu.
\newblock Critical points of multidimensional random Fourier series: variance estimates
\newblock  Preprint, \url{http://arxiv.org/abs/1310.5571}



\bibitem{Rice1}
S.~O. Rice.
\newblock The {D}istribution of the {M}axima of a {R}andom {C}urve.
\newblock {\em Amer. J. Math.}, 61(2):409--416, 1939.

\bibitem{Rice2}
S.~O. Rice.
\newblock The {D}istribution of the {M}axima of a {R}andom {C}urve.
\newblock {\em {S}elected {P}apers on {N}oise and {S}tochastic {P}rocesses},
  pages 409--416, 1954.

\bibitem{ShubSmale}
M.~Shub and S.~Smale.
\newblock Complexity of {B}ezout's theorem. {II}. {V}olumes and probabilities.
\newblock In {\em Computational algebraic geometry ({N}ice, 1992)}, volume 109
  of {\em Progr. Math.}, pages 267--285. Birkh\"auser Boston, Boston, MA, 1993.

\bibitem{TaoVu}
Terence Tao and Van Vu.
\newblock Local universality of zeroes of random polynomials.
\newblock To appear in {\it International Math Research Notices}; see also
  \url{http://arxiv.org/abs/1307.4357v1}.

\end{thebibliography}
\end{document}